\documentclass[12pt]{article}
\usepackage{amsfonts}
\usepackage{amsmath}
\usepackage{amssymb}
\usepackage{graphicx}
\usepackage{color}
\usepackage[all, knot]{xy}
\usepackage{tikz}
\usepackage{ulem}
\usepackage{array}

\usepackage[utf8]{inputenc}
\usepackage{epstopdf}
\usepackage[footnotesize]{caption}
\usepackage{amsthm}
\usepackage{enumitem}
\usepackage{mathrsfs}

\usepackage[margin=3cm]{geometry}

\def \be {\begin{equation}}
\def \ee {\end{equation}}
\def \bea {\begin{eqnarray}}
\def \eea {\end{eqnarray}}
\def \nn {\nonumber}

\def \rr {\raise.35ex\hbox{\small $\prime$}\kern-.17em{\mbox{\large $\imath$}}}

\def \dels {\partial\kern-.6em /\kern.1em}
\def \As {{A\kern-.5em / \kern.5em}}
\def \Ds {D\kern-.7em / \kern.5em}

\def \ks {k\kern-.5em /}
\def \ls {l\kern-.5em /}

\newcommand{\ci}[1]{}
\newcommand{\ba}{\begin{eqnarray}}
\newcommand{\ea}{\end{eqnarray}}
\newcommand{\bal}{\begin{align}}
\newcommand{\eal}{\end{align}}
\newcommand{\bay}[1]{\left(\begin{array}{#1}}
\newcommand{\eay}{\end{array}\right)}

\newcommand{\ket}[1]{|{#1}\rangle}
\newcommand{\bra}[1]{\langle{#1}|}

\newcommand{\hide}[1]{}

\newtheorem{theorem}{Theorem}

\newtheorem{lemma}{Lemma}

\newlist{axioms}{enumerate}{2}
\setlist[axioms,1]{label=\textbf{A\arabic{axiomsi}.}, ref=A\arabic{axiomsi}}
\setlist[axioms,2]{label=\textbf{A\arabic{axiomsi}\rlap{\myEnumCounter{axiomsii}}.},%
                   ref=A\arabic{axiomsi}\myEnumCounter{axiomsii},%
                   align=parleft,%
                   leftmargin=0em,%
                   itemsep=1.4ex,%
                   before={\stepcounter{axiomsi}}}
                   
 \usetikzlibrary{decorations.markings}

\begin{document}

\begin{titlepage}

\begin{center}

\hfill
\vskip .2in

\textbf{\LARGE
Bell's Inequality, Generalized Concurrence and Entanglement in Qubits
\vskip.3cm
}

\vskip .5in
{\large
Po-Yao Chang$^a$ \footnote{e-mail address: pychang@pks.mpg.de}, Su-Kuan Chu$^{b,c}$ \footnote{e-mail address: skchu@terpmail.umd.edu}, and Chen-Te Ma$^{d,e,f}$ \footnote{e-mail address: yefgst@gmail.com} 
\\
\vskip 1mm
}
{\sl
$^a$
Center for Materials Theory, Rutgers University, Piscataway, New Jersey 08854, USA.
\\
$^b$
Joint Quantum Institute, NIST/University of Maryland, College Park,\\
 Maryland 20742, USA.
\\
$^c$
Joint Center for Quantum Information and Computer Science,\\
NIST/University of Maryland, College Park, Maryland 20742, USA.
\\
$^d$
School of Physics and Telecommunication Engineering,\\
South China Normal University, Guangzhou 510006, China.
\\
$^e$
The Laboratory for Quantum Gravity and Strings,\\
Department of Mathematics and Applied Mathematics,\\
University of Cape Town,Private Bag, Rondebosch 7700, South Africa.
\\
$^f$
Department of Physics and Center for Theoretical Sciences, \\
National Taiwan University, Taipei 10617, Taiwan, R.O.C..
}\\
\vskip 1mm
\vspace{40pt}
\end{center}
\newpage
\begin{abstract}
It is well known that the maximal violation of the Bell's inequality for a two-qubit system is related
to the entanglement formation in terms of a concurrence.
However, a generalization of this relation to an $n$-qubit state has not been found.
In the paper, we demonstrate some extensions of 
the relation between the upper bound of the Bell's violation
and a {\it generalized} concurrence in several $n$-qubit states.
In particular, we show the upper bound of the Bell's violation can be expressed as a function of the generalized concurrence,
if a state can be expressed in terms of two variables.
We apply the relation to the Wen-Plaquette model and show that the topological entanglement entropy can be extracted from the maximal Bell's violation.
\end{abstract}

\end{titlepage}

\section{Introduction}
\label{1}
Quantum entanglement is an essential concept in quantum systems that has no classical counterpart. 
In the past five decades, many developments including quantum information \cite{Steane:1997kb}, quantum algorithm \cite{Galindo:2001ei} and Bell's inequality \cite{Clauser:1969ny, Bell:1964kc},
provide a deeper understanding of the quantum world. The basic concept of the quantum entanglement 
can be understood from the separability of a quantum state, which states if a state is inseparable, it is called an ``entangled state''.
The Bell state is the simplest example of an entanglement state $|\psi\rangle_{\rm Bell} = \frac{1}{\sqrt{2}}( |0\rangle_A \otimes |1 \rangle_B- |1\rangle_A \otimes |0 \rangle_B) $,
where the $|0(1)\rangle_{A(B)}$ represents a two level state in the subsystem $A (B)$.

Although the definition of an entangled state is clear, how to measure the "entanglement" is rather subtle.
On the one hand,
the most common measure of entanglement is given by the entanglement entropy of a region $A$,
$ S_{\mathrm{EE}, A}=-{\rm Tr } \rho_A \ln \rho_A$ with 
$ \rho_A={\rm Tr_B \rho}$
  being a reduced density matrix of a subsystem $A$ and $\rho$ being a density matrix of a Hilbert space
  $\mathbb{H}=\mathbb{H}_A\otimes\mathbb{H}_B$.
When the entanglement entropy vanishes, the quantum entanglement between complementary subsystems $A$ and $B$ disappears. 
This measure of entanglement only requires the information of the local (reduced) density matrix $\rho_A$, which encodes
its degree of entanglement with the complementary subsystem $B$.
On the other hand, quantum entanglement is also encoded in the correlations between two local measurements on subsystems $A$ and $B$. The famous result of the violation of Bell's inequality \cite{Clauser:1969ny, Bell:1964kc} states that 
the correlations between different measurements of two separated particles of an entangled state must satisfy the inequality under the local realism.
The violation of the constraints or Bell's inequality indicates quantum entanglement in the system, 
which is demonstrated in two-qubit systems both theoretically \cite{Cirelson:1980ry} and experimentally \cite{Hensen:2015ccp, Giustina:2015, Shalm:2015}. 

Although the violation of the Bell's inequality indicates quantum entanglement, how to quantify the "entanglement" from the
Bell's correlators in many-body systems are not clear. 
In a two-qubit system, a relation between entanglement entropy, measured in terms of the concurrence \cite{Bennett:1996gf}, 
and violation of the Bell's inequality is shown\cite{Verstraete:2001} through eigenvalues of an $R$-matrix \cite{Horodecki:1995}. An upper bound of Bell's inequality for some three-qubit quantum states was also studied \cite{Kaiewski:2016}.
These studies show the entanglement, in terms of purity and concurrence, can be captured from the extremal Bell violations in few-body systems.
In this paper, we extended this concept further to $n$-qubit systems.
In particular, 
we consider an $n$-qubit state with two variables $(\lambda_+,\lambda_-)$ such that the concurrence can be expressed as
$C(\psi(\lambda_+,\lambda_-))$. 
We show that if there exists a inverse mapping $(\lambda_+(C(\psi)),\lambda_-(C(\psi)))$, 
the upper bound of the violation of the Bell's inequality measured from an $n$-qubit Bell's operator \cite{Gisin:1998ze, Chang:2017czx}
 can be expressed a function of a {\it generalized} concurrence.
In addition, we show in Theorem 2.1, this upper bound is the maximal violation of the Bell's inequality for a specific $n$-qubit state.

The application of this relation in $n$-qubit systems is tremendous. 
In cutting edge experiments, quantum entanglement can be measured beyond a qubit system but is only accessible for few-body systems. This is because of the measures of entanglement
requires probes with a single-site resolution such as a beam splitter operation in ultracold atoms \cite{Kaufman:2016} or 
the nuclear magnetic resonance (NMR) quantum simulator on molecular \cite{Li2016, Luo:20161, Peng2010, Tseng1999, Luo:20162}.
On the other hand, the measure of Bell's operators is accessible for many-body systems and is demonstrated experimentally for 480 atoms in a Bose-Einstein condensate\cite{Schmied:2016}.
The relation between the upper bound of the maximal violation of the Bell's inequality and a generalized concurrence in $n$-qubit systems 
may provide an alternative measure of entanglement beyond few-body systems. 
The application of this generalization is the essential motivation presented in this paper.

Although we did not find a general relation for a generic $n$-qubit state, the examples that we demonstrated can apply to topological orders. Simple $n$-qubit models of topological orders are the toric code model \cite{Kitaev:1997wr} and the Wen-Plaquette model \cite{Wen:2003yv},  which are constructed from stabilizer operators.
It was shown that the topological entanglement entropy can detect the number of distinct quasi-particles when subregions are contractible \cite{Kitaev:2005dm,Levin:2006zz}. For a non-contractible region, one can also obtain the topological entanglement entropy, but it does not imply a number of distinct quasi-particles \cite{Grover:2013ifa}.
Hence, the entanglement measurement can be a direct probe of detecting topological orders. We show the topological entanglement entropy
can be measured from the maximal violation of the Bell's inequality in a Wen-Plaquette model \cite{Wen:2003yv}.

The structure of this paper is as follows. In Sec. \ref{2}, we study a relation between entanglement entropy and maximal violation of Bell's inequality in the
case that 
each quantum state is a linear combination of two product states.
In Sec. \ref{3}, we demonstrate this relation in an XY model at zero and finite temperature and its application in the Wen-Plaquette model. 
We show the topological entanglement entropy in the Wen-Plaquette model can be computed from
the relation of the maximal violation of the Bell's inequality and the generalized concurrence. 
In Sec. \ref{4}, we demonstrate that this relation is also hold in a certain type of
2$n$-qubit states. Finally, we discuss and conclude in Sec. \ref{5}.

\section{Entanglement and Maximal Violation}
\label{2}
The Bell's operator of $n$ qubits is defined iteratively as that \cite{Gisin:1998ze}
\bea
{\cal B}_n={\cal B}_{n-1}\otimes\frac{1}{2}\bigg(A_n+A_n^{\prime}\bigg)+{\cal B}^{\prime}_{n-1}\otimes\frac{1}{2}\bigg(A_n-A_n^{\prime}\bigg),
\eea
where 
$A_n={\bf a}_n \cdot \boldsymbol{\sigma}$, and  $A^{\prime}_n={\bf a}_n^{\prime} \cdot \boldsymbol{\sigma}$
are the operators in the $n$-th qubit
with ${\bf a}_n$ and ${\bf a}_n^{\prime}$ being unit vectors and 
$\boldsymbol{\sigma}=(\sigma_x,\sigma_y,\sigma_z)$
 being a vector of the Pauli matrices. The $(n-1)$-qubit operators $\frac{1}{2} {\cal B}_{n-1}$ and $\frac{1}{2} {\cal B}_{n-1}^{\prime}$ act on the rest of the qubits.
Noe that we choose: 
\bea
\frac{1}{2}{\cal B}_1={\bf b} \cdot \boldsymbol{\sigma}, \qquad \frac{1}{2}{\cal B}_1^{\prime}={\bf b}^{\prime} \cdot \boldsymbol{\sigma}
\eea
with ${\bf b}$ and ${\bf b}^{\prime}$ being unit vectors.
It is known that for an $n$-qubit system, the upper bound of the expectation value of 
the Bell's operator \cite{Gisin:1998ze}
\bea
\mbox{Tr}(\rho{\cal B}_n)\le 2^{\frac{n+1}{2}}
\eea 
leads to violation of the Bell's inequality \cite{Clauser:1969ny}.

For a given density matrix, the maximal expectation value of a Bell's operator is called the maximal violation of Bell's inequality. 
Here we prove a relation between maximal violation of the Bell's inequality and a concurrence of a pure state or entanglement entropy in an $n$-qubit system \cite{Chang:2017czx} when the all $i$-th operators in the Bell's operator are $A_i$ and $A_i^{\prime}$ for $2\le i< n$
\bea
\tilde{{\cal B}}_n
=&&{\cal B}_1\otimes A_2\otimes A_3\cdots\otimes A_{n-2}\otimes A_{n-1}\otimes\frac{1}{2}\bigg(A_n+A_n^{\prime}\bigg)
\nn\\
&&+{\cal B}^{\prime}_1\otimes A_2^{\prime}\otimes A_3^{\prime}\cdots\otimes A_{n-2}^{\prime}\otimes A_{n-1}^{\prime}
\otimes\frac{1}{2}\bigg(A_n-A_n^{\prime}\bigg).
\label{Eq:mBell}
\eea

To proceed our derivation, we first introduce the generalized $R$-matrix:
\bea
R_{i_1i_2\cdots i_n}\equiv\mbox{Tr}(\rho\sigma_{i_1}\otimes\sigma_{i_2}\otimes\cdots\otimes\sigma_{i_n})\equiv R_{Ii_n},
\label{Eq:RM}
\eea
where $\rho$ is a density matrix, $\sigma_{i_\alpha}$ is the Pauli matrix labeled by $i_\alpha=x,y,z$ with site indices $\alpha=1, 2, \cdots, n$. We express the generalized $R$-matrix as a $3^{n-1} \times 3$ matrix $R_{I i_n}$ with the first index being a multi-index $I=i_1i_2\cdots i_{n-1}$ and the second index being $i_n$. In a two-qubit system, maximal violation of the Bell's inequality
was computed from the $3 \times 3$ matrix $R_{ij}$ \cite{Horodecki:1995}, which is a special case of the generalized $R$-matrix or the two-qubit system $n=2$.
Now we use the generalized $R$-matrix to generalize the maximal violation of the Bell's inequality ($\tilde{{\cal B}}_n$) in an $n$-qubit system \cite{Chang:2017czx}.
\begin{lemma}
\label{me}
The maximal violation of the Bell's inequalities has the following relation:
\bea
\gamma\equiv\max_{\tilde{{\cal B}}_n} {\rm Tr}(\rho{\tilde{\cal B}}_n)\le2\sqrt{u_1^2+u_2^2},
\eea
where $u^2_1$ and $u^2_2$ are the first two largest eigenvalues of the matrix $R^\dagger R$ when a number of qubits larger than two, $n>2$, and 
\bea
\gamma=2\sqrt{u_1^2+u_2^2}
\label{Eq7}
\eea
when a number of qubits is two, $n=2$. Note that the matrix $R^{\dagger}R$ is contracted over the multi-index $I$.
\end{lemma}
\begin{proof}
We first introduce two three-dimensional orthonormal vectors $\hat{c}$ and $\hat{c'}$ such that:
\bea
\hat{a}+\hat{a'}=2\hat{c} \cos \theta, \qquad \hat{a}-\hat{a'}=2 \hat{c'} \sin \theta,
\eea
 where $\theta\in\lbrack 0,\frac{1}{2}\pi\rbrack$, through three-dimensional unit vectors $\hat{a}$ and $\hat{a'}$. The maximal violation of the Bell's inequality is defined as 
 \bea
 \gamma\equiv\max_{\tilde{{\cal B}}_n} \mbox{Tr}(\rho\tilde{{\cal B}}_n)
 \eea
with the Bell's operator of $n$-qubit $\tilde{{\cal B}}_n$ defined in \eqref{Eq:mBell}.
By using the generalized $R$-matrix and the unit vectors:
\bea
\hat{B}&\equiv&\hat{B}_I=\hat{B}_{i_1i_2\cdots i_{n-1}}\equiv\hat{a}_{1, i_1}\hat{a}_{2, i_2}\cdots\hat{a}_{n-1, i_{n-1}}, 
\nn\\
\hat{B'}&\equiv&\hat{B'}_I=\hat{B'}_{i_1i_2\cdots i_{n-1}}\equiv\hat{a'}_{1, i_1}\hat{a'}_{2, i_2}\cdots\hat{a'}_{n-1, i_{n-1}},
\nn\\
\hat{a}&\equiv&\hat{a}_{n, i_n}, 
\quad
\hat{a'}\equiv\hat{a^{\prime}}_{n, i_n},
\eea 
in which ${\hat{B}}$ and $\hat{B^{\prime}}$ are unit vectors in $3^{n-1}$ dimensions, we have:
\bea
\gamma
&=&\max_{\hat{B},\hat{B'},\hat{a},\hat{a'}}  \bigg(\langle\hat{B}, R (\hat{a}+\hat{a'})\rangle+\langle\hat{B'}, R (\hat{a}-\hat{a'})\rangle\bigg)   
\nn\\
&\le&\max_{\hat{c}, \hat{c'},\theta} \bigg( 2||R \hat{c}|| \cos \theta +2||R \hat{c'}|| \sin \theta\bigg)  
= 2\sqrt{u_1^2+u_2^2},
\eea 
in which $u^2_1$ and $u^2_2$ are the first two largest eigenvalues of the matrix $R^\dagger R$.
The inner product and the norm are defined as:
\bea
\langle P, Q\rangle\equiv P^{\dagger}Q, \qquad ||U||\equiv \sqrt{U^{\dagger}U}.
\eea
 Because $R(\hat{a}+\hat{a^{\prime}})$ and $R(\hat{a}-\hat{a^{\prime}})$ are defined in $3^{n-1}$ dimensions and each unit vector $\hat{B}$ and $\hat{B^{\prime}}$ only contains $2(n-1)$ parameters, there is no guarantee for the below relations:
 \bea
 \hat{B}=k_1R(\hat{a}+\hat{a^{\prime}}), \qquad \hat{B^{\prime}}=k_2R(\hat{a}-\hat{a^{\prime}}),
 \eea
except for $n=2$, where $k_1$ and $k_2$ are two arbitrary constants.
\end{proof}
An earlier approach to relate the maximal violation of the Bell's inequality
and the concurrence of a pure state \cite{Bennett:1996gf}
\bea
C(\psi)\equiv\sqrt{2(1-{\rm Tr} \rho_A^2)}
\label{Eq:C}
\eea 
in a two-qubit system was discussed \cite{Verstraete:2001}.

Now we generalize the relation of the maximal violation of the Bell's inequality
and the concurrence of a pure state in an $n$-qubit system, which still monotonically increases with respect to entanglement entropy, when a quantum state is a linear combination of two product states. The concurrence of the pure state is computed with respect to the bipartition with ($n-1$) qubits in a subsystem $B$ and one qubit in a subsystem $A$.
\begin{theorem}
\label{mec}
We consider an $n$-qubit state 
\bea
|\psi \rangle=|u\rangle_B\otimes\big(\lambda_+|v\rangle_B \otimes |1\rangle_A +\lambda_-|\tilde{v}\rangle_B \otimes |0\rangle_A\big)
\label{Eq:F1}
\eea
with $\lambda_+|v\rangle_B \otimes |1\rangle_A +\lambda_-|\tilde{v}\rangle_B \otimes |0\rangle_A$ being a non-biseparable, $\alpha$-qubit state,
$|u\rangle_B$, $|v\rangle_B$, $|\tilde{v}\rangle_B$ being product states consisting of $\ket{0}$'s and $\ket{1}$'s. The state $|v\rangle_B$ is orthogonal to the state $|\tilde{v}\rangle_B$ by choosing opposite bits on each site. Coefficients $\lambda_+$ and $\lambda_-$ are real numbers and $\lambda_+^2+\lambda_-^2=1$.
The maximal violation of the Bell's inequality in an $n$-qubit system is 
\bea
\gamma=2f_{\alpha}(\psi),
\label{Eq:T1}
\eea
in which the function $f_{\alpha}(\psi)$ is defined as:\\
$(1)$ $\alpha$  is an even number:
\begin{align}
&f_{\alpha}(\psi)\equiv\sqrt{1+2^{\alpha-2}C^2(\psi)},& \qquad    &2^{2-\alpha}\ge C^2(\psi), \nn\\
&f_{\alpha}(\psi)\equiv2^{\frac{\alpha-1}{2}}C(\psi),& \qquad       &2^{2-\alpha}\le C^2(\psi).
\label{Eq:f1}
\end{align}
$(2)$ $\alpha$  is an odd number:
\begin{align}
&f_{\alpha}(\psi)\equiv\sqrt{1+\big(2^{\alpha-2}-1\big)C^2(\psi)},& \qquad  &\frac{1}{1+2^{\alpha-2}}\ge C^2(\psi),\nn\\
&f_{\alpha}(\psi)\equiv2^{\frac{\alpha-1}{2}}C(\psi),& \qquad &\frac{1}{1+2^{\alpha-2}}\le C^2(\psi).
\label{Eq:f2}
\end{align}
Here $C(\psi)$ is the concurrence of the pure state computed with respect to the bipartition that the subsystem $B$ contains $(n-1)$ qubits and the subsystem $A$ contains one qubit. 
\end{theorem}
\begin{proof}
The Hilbert space of an $n$-qubit system is bipartitioned as 
$\mathbb{H}=\mathbb{H}_B \otimes \mathbb{H}_A$,
in which dimensions of the sub-Hilbert spaces are
$\dim(\mathbb{H}_A)=2$ and $\dim(\mathbb{H}_B)=2^{n-1}$.
We consider a quantum state with respect to this bipartition 
\bea
|\psi \rangle=|u\rangle_B\otimes(\lambda_+|v\rangle_B \otimes |1\rangle_A +\lambda_-|\tilde{v}\rangle_B \otimes |0\rangle_A),
\eea
where $|u\rangle_B\otimes|v\rangle_B$ and $|u\rangle_B\otimes|\tilde{v}\rangle_B$ are the product states in $\mathbb{H}_B$
and $|1\rangle_A$ and $|0\rangle_A$ are the product states in $\mathbb{H}_A$. By using the properties:
\bea
{\rm Tr}\rho_A=\lambda_+^2 +\lambda_-^2=1, 
\quad
C(\psi)=\sqrt{2(1-\lambda_+^4-\lambda_-^4)},
\eea 
the coefficients $\lambda_{\pm}$ can be expressed in terms of the concurrence
\bea
\lambda_{\pm}^2=\big(1\pm\sqrt{1-C^2(\psi)}\big)/2.
\label{EQ 28}
\eea
The generalized $R$-matrix is:
\bea
R_{Ix}&=&\lambda_+\lambda_- {\rm Tr}  \left [  \sigma_{I_ 1}| u \rangle \langle u |     
 \otimes \sigma_{I_ 2} (|v \rangle \langle \tilde{v}|+|\tilde{v} \rangle \langle v|)  \right],  \notag\\
R_{Iy}&=&-i \lambda_+\lambda_- {\rm Tr}  \left [  \sigma_{I_ 1}| u \rangle \langle u |     
\otimes \sigma_{I_ 2} (|v \rangle \langle \tilde{v}|-|\tilde{v} \rangle \langle v|)  \right],  \notag\\
R_{Iz}&=& -\lambda_+^2{\rm Tr}  \left [ \sigma_{I_ 1}| u \rangle \langle u |     
\otimes \sigma_{I_ 2} |v \rangle \langle {v}|  \right]  +\lambda_-^2{\rm Tr}  \left [  \sigma_{I_ 1}| u \rangle \langle u |     
\otimes \sigma_{I_ 2} |\tilde{v} \rangle \langle \tilde{v}|  \right],
\eea
where $I\equiv I_1I_2$ concatenating two strings of indices,
$I_1\equiv i_1\cdots i_{n-\alpha-1}$ and $I_2\equiv i_{n-\alpha}\cdots i_{n-1}$.
We also define that
$\sigma_{I_1} \equiv \sigma_1\otimes \cdots \otimes \sigma_{i_{n-\alpha-1}}$, 
$\sigma_{I_2} \equiv \sigma_{n-\alpha} \otimes \cdots \otimes \sigma_{i_{n-1}}$.

 Here we choose the basis, 
 $|0\rangle\equiv (1,0)^{\rm T}$ and $|1\rangle\equiv  (0,1)^{\rm T}$.
  One should notice that the dimensions of the matrix $|u\rangle \langle u|$ is  $2^{n-\alpha}$
and the dimensions of the matrices
$|v\rangle \langle v |$,   $|\tilde{v} \rangle \langle \tilde{v}|$, $|v\rangle \langle \tilde{v}|$,  $|\tilde{v}\rangle \langle v |$ 
are $2^{\alpha-1}$. 
The non-vanishing matrix elements of the generalized $R$-matrix $R_{I\alpha}$, $\alpha= x, y, z$, come from the diagonal matrix elements of the below matrices: 
\bea
&&\sigma_{I_1}  |{u}\rangle \langle u| \otimes \sigma_{I_ 2} |{v} \rangle \langle {v} |, \qquad
 \sigma_{I_1}  |{u}\rangle \langle u| \otimes \sigma_{I_ 2} |\tilde{v} \rangle \langle \tilde{v} |, 
 \nn\\
 &&\sigma_{I_1}  |{u}\rangle \langle u| \otimes \sigma_{I_ 2} |\tilde{v} \rangle \langle {v} |, \qquad
 \sigma_{I_1}  |{u}\rangle \langle u| \otimes \sigma_{I_ 2} |{v} \rangle \langle \tilde{v} |.
 \eea
Then the conditions of non-vanishing matrix elements of the generalized $R$-matrix $R_{I\alpha}$, where $\alpha= x, y, z$,
require the followings:
\bea
 \sigma_{I_1}  |{u}\rangle  \to  |{u}\rangle,\qquad
 \sigma_{I_ 2} |\tilde{v} \rangle \to |{v} \rangle, \qquad \sigma_{I_ 2} |{v} \rangle \to |\tilde{v} \rangle 
 \eea
  for the generalized $R$-matrix $R_{Ix(y)}$
and 
\bea
\sigma_{I_ 2} |u \rangle \to |u \rangle, \qquad \sigma_{I_ 2} |{v} \rangle \to |{v} \rangle, \qquad
 \sigma_{I_ 2} |\tilde{v} \rangle \to |\tilde{v} \rangle
 \eea
  for the generalized $R$-matrix $R_{Iz}$.

The conditions of the non-vanishing matrix elements $R_{Ix}$ are
$(n-\alpha)$  number of the $\sigma_z$ matrices in the indices $I_1$, $(\alpha-1-j_1)$ number of the $\sigma_x$ matrices and $j_1$ number of the $\sigma_y$
matrices in the indices $I_2$ with $j_1$ being an even integer.
The conditions of the non-vanishing matrix elements $R_{Iy}$ are
$(n-\alpha)$  number of the $\sigma_z$ matrices in the indices $I_1$, $(\alpha-1-j_2)$ number of the $\sigma_x$ matrices and $j_2$ number of the $\sigma_y$
matrices in indices $I_2$ with $j_2$ being an odd integer.
The conditions of the non-vanishing matrix elements $R_{Iz}$ are
$(n-\alpha)$  number of the $\sigma_z$ matrices in the indices $I_1$, $(\alpha-1)$ number of the $\sigma_z$ 
matrices in the indices $I_2$. 

The above conditions lead to a diagonal form of the matrix $R^\dagger R$.
In the case that $\alpha$ is an even integer, the eigenvalues of the matrix $R^\dagger R$ are:
\bea
&&\big(1+C^{n-1}_2+C^{n-1}_4+\cdots+C^{n-1}_{n-1}\big)C^2(\psi_i)=2^{n-2}C^2(\psi_i), 
\nn\\
 &&\big(1+C^{n-1}_2+C^{n-1}_4+\cdots+C^{n-1}_{n-1}\big)C^2(\psi_i)=2^{n-2}C^2(\psi_i), 
\nn\\
&&1.
\eea
In the case that the integer $\alpha$ is an odd integer, the eigenvalues of the matrix $R^\dagger R$  are:
\bea
&&\big(1+C^{n-1}_2+C^{n-1}_4+\cdots+C^{n-1}_{n-1}\big)C^2(\psi_i)=2^{n-2}C^2(\psi_i), 
\nn\\
 &&\big(1+C^{n-1}_2+C^{n-1}_4+\cdots+C^{n-1}_{n-1}\big)C^2(\psi_i)=2^{n-2}C^2(\psi_i), 
\nn\\
  &&1-C^2(\psi_i).
\eea
We used $C^n_k=C^{n-1}_{k-1}+C^{n-1}_{k-2}$ , which comes from each coefficient of the equation $(1+x)^n=(1+x)^{n-1}(1+x)$, to compute the $xx$-component of the matrix $(R^{\dagger}R)_{xx}$ and the $yy$-component of the matrix $(R^{\dagger}R)_{yy}$.

From the eigenvalues of the matrix $R^{\dagger}R$, we obtain the followings:
\begin{align}
&f_{\alpha}(\psi)\equiv\sqrt{1+2^{\alpha-2}C^2(\psi)},& \qquad    &2^{2-\alpha}\ge C^2(\psi), \nn\\
&f_{\alpha}(\psi)\equiv2^{\frac{\alpha-1}{2}}C(\psi),& \qquad       &2^{2-\alpha}\le C^2(\psi)
\end{align}
when $\alpha$ is an even number,
\begin{align}
&f_{\alpha}(\psi)\equiv\sqrt{1+\big(2^{\alpha-2}-1\big)C^2(\psi)},& \qquad  &\frac{1}{1+2^{\alpha-2}}\ge C^2(\psi),\nn\\
&f_{\alpha}(\psi)\equiv2^{\frac{\alpha-1}{2}}C(\psi),& \qquad &\frac{1}{1+2^{\alpha-2}}\le C^2(\psi)
\end{align}
when $n$ is an odd number and the inequality
\bea
\gamma\le 2f_{\alpha}(\psi).
\eea

Now we want to show that the maximal violation of the Bell's inequality should satisfy the below equations:
\bea
\gamma = \max_{\hat{B},\hat{B'},\hat{a},\hat{a'}}  \langle\hat{B}, R (\hat{a}+\hat{a'})\rangle+\langle\hat{B'}, R (\hat{a}-\hat{a'})\rangle   
= 2\sqrt{u_1^2+u_2^2},
\eea 
in which $u^2_1$ and $u^2_2$ are the first two largest eigenvalues of the matrix $R^\dagger R$ and the variables are defined as the followings:
\bea
\hat{B} & \equiv & \hat{a}_{1, i_1}\hat{a}_{2, i_2}\cdots\hat{a}_{n-1, i_{n-1}},
\quad
\hat{B'} \equiv \hat{a'}_{1, i_1}\hat{a'}_{2, i_2}\cdots\hat{a'}_{n-1, i_{n-1}}, 
\nn\\
\hat{a}&\equiv&\hat{a}_{n, i_n}
\quad
\hat{a^{\prime}}\equiv\hat{a^{\prime}}_{n, i_n},
\eea
 where 
 \bea
 \hat{a}_{n, i_n}+\hat{a^{\prime}}_{n ,i_n} \equiv 2\hat{c}_{n, i_n}\cos\theta, 
\quad
  \hat{a}_{n, i_n}-\hat{a^{\prime}}_{n ,i_n} \equiv 2\hat{c^{\prime}}_{n, i_n}\sin\theta, 
 \eea 
 with $\theta\in\lbrack 0, \pi/2\rbrack$.
This equality holds when the below relations are satisfied:
\bea
\hat{B}=k_1R(\hat{a}+\hat{a^{\prime}}), \qquad \hat{B'}=k_2R(\hat{a}-\hat{a^{\prime}}),
\eea
where $k_1$ and $k_2$ are constants.
One natural choice of $a_{\alpha,  i_{\alpha}}$ and $a_{\alpha,^\prime  i_{\alpha}}$ can be obtained by equating two ratios:
\bea
\bigg|\frac{R_{Ix}(\hat{a}_x + \hat{a^{\prime}}_x)}{R_{I'y}(\hat{a}_y+\hat{a^{\prime}}_y)}\bigg|=\bigg|\frac{B^{}_I}{B^{}_{I'}}\bigg|, 
\quad
\bigg|\frac{R_{Ix}(\hat{a}_x-\hat{a^{\prime}}_x)}{R_{I'y}(\hat{a}_y-\hat{a^{\prime}}_y)}\bigg|=\bigg|\frac{B^{\prime}_I}{B^{\prime}_{I'}}\bigg|,
\eea 
where the multi-index $I$ and the multi-index $I'$ are chosen in the way that one site of the index $I_2$ in the multi-index $I$ is labeled by $x$ and in the multi-index $I'$ is labeled by $y$, and other sites of the indices $I_2$ in the multi-index $I$ and the multi-index $I^{\prime}$ are labeled by the same symbols.
This leads to the following equations:
\bea
\bigg|\frac{\hat{a}_{I, x}}{\hat{a}_{I^{\prime}, y}}\bigg|=\bigg|\frac{\hat{c}_{n, x}}{\hat{c}_{n, y}}\bigg|, \qquad \bigg|\frac{\hat{a^{\prime}}_{I, x}}{\hat{a^{\prime}}_{I^{\prime}, y}}\bigg|=\bigg|\frac{\hat{c^{\prime}}_{n, x}}{\hat{c^{\prime}}_{n, y}}\bigg|.
\eea
When the eigenvalues are
\bea
u^2_1=(R^{\dagger}R)_{xx}, \qquad u^2_2=(R^{\dagger}R)_{yy},
\eea
 we can choose the followings:
\bea
&&(\hat{c}_{n, x}, \hat{c}_{n, y}, \hat{c}_{n, z})^{\rm T} = \frac{1}{\sqrt{2}}(1, 1, 0), 
\quad
 (\hat{c^{\prime}}_{n, x}, \hat{c^{\prime}}_{n, y}, \hat{c^{\prime}}_{n, z})^{\rm T} = \frac{1}{\sqrt{2}}(1, -1, 0), 
\nn\\
&&\cos(\theta)=\sin(\theta)=\frac{\sqrt{2}}{2},
\eea
 to show that the maximal violation of the Bell's inequality saturates the upper bound
 \bea
 \gamma=2\sqrt{u_1^2+u_2^2}.
 \eea
  For the other case:
  \bea
  u^2_1=(R^{\dagger}R)_{zz}, \qquad u^2_2=(R^{\dagger}R)_{xx},
  \eea
   we can choose the followings:
   \bea
&&(\hat{c}_{n, x}, \hat{c}_{n, y}, \hat{c}_{n, z})^{\rm T} = (0, 0, 1), 
(\hat{c^{\prime}}_{n, x}, \hat{c^{\prime}}_{n, y}, \hat{c^{\prime}}_{n, z})^{\rm T} = \frac{1}{\sqrt{2}}(1, 1, 0),
\nn\\
&&\cos(\theta) = \frac{u_1}{\sqrt{u_1^2+u_2^2}}, 
\quad
\sin(\theta) = \frac{u_2}{\sqrt{u_1^2+u_2^2}}
\eea
 to prove that the maximal violation of the Bell's inequality saturates the upper bound
 \bea
 \gamma=2\sqrt{u_1^2+u_2^2}.
 \eea
 \\
\end{proof}
We prove that the maximal violation of the Bell's inequality ($\tilde{\cal{B}}_n$) is directly related to the concurrence of the pure state when the subsystem $A$ only contains one qubit and the quantum state is a linear combination of two product states. Because the concurrence of the pure state also monotonically increases with respect to entanglement entropy, the maximal violation of the Bell's inequality is also related to entanglement entropy directly.
For the maximally entangled state with the maximal concurrence of the pure state
$ C(\psi)=1$, 
the maximal violation of the Bell's inequality is 
$\gamma=2^{\frac{\alpha+1}{2}}\leq 2^{\frac{n+1}{2}}$
satisfies the upper bound of the Bell's operator in an $n$-qubit system. Although we do not use the most generic form of the Bell's operator, information of the quantum state is already contained in the $n$-th qubit operators. Thus, the computing of the generalized $R$-matrix should give the maximal violation of the Bell's inequality ($\tilde{\cal{B}}_n$) directly when a quantum state is a linear superposition of two product states.

Now we discuss the maximal violation of the Bell's inequality in a mixed state. The mixed state of a density matrix $\rho$ is
\bea
\rho=\sum_i p_i|\psi_i\rangle\langle\psi_i|,
\eea
in which the sum of $p_i$ is one. Because we have that:
\bea
\gamma(\rho)&:=&\max_{{\cal B}_n}\mbox{Tr}\big(\rho{\cal B}_n\big)=\max_{{\cal B}_n}\mbox{Tr}\bigg(\sum_ip_i|\psi_i\rangle\langle \psi_i|{\cal B}_n\bigg)
\nn\\
&=&\max_{{\cal B}_n}\sum_ip_i\langle \psi_i|{\cal B}_n|\psi_i\rangle
\le \sum_i p_i\cdot \max_{{\cal B}_n}\langle\psi_i|{\cal B}_n|\psi_i\rangle,
\label{Eq:mixed}
\eea

We find the upper bound of the maximal violation of the Bell's inequality of the mixed state is the sum
of the maximal violation of the Bell's inequality of the pure state with corresponding weight $p_i$.

Next, we want to relate this upper bound of the maximal violation of the Bell's inequality of a mixed state to the concurrence of a mixed state.
From Eq. (\ref{Eq:mixed}), we have
\begin{align}
\gamma(\rho) \le 2 \sum_i p_i  f_\alpha(C(\psi_i))  ,
\end{align}
where  $f_\alpha(C(\psi_i))$ is defined in Eqs. (\ref{Eq:f1}) (\ref{Eq:f2}).
The concurrence for mixed state is defined as \cite{Wootters:1997id}
\bea
C(\rho)\equiv\min_{p_i, \psi_i}\sum_i p_iC(\psi_i), 
\label{Eq:C_mixed}
\eea
where $\{p_i, \psi_i\}$ is a chosen decomposition of $\rho$ such that minimize the right hand side of Eq. (\ref{Eq:C_mixed}).
When the function $f_{\alpha}$ is a linear function of the concurrence $C(\psi_i)$, $f_\alpha(C(\psi_i)) \propto C(\psi_i)$, we have the equality $\min_{p_i, \psi_i} \sum_{i} p_i f_\alpha(C(\psi_i)) = f_\alpha(\min_{p_i, \psi_i} \sum_i p_iC(\psi_i)) = f_\alpha(C(\rho))$.
The maximal violation of the Bell's inequality of a mixed state is bounded by the function of the concurrence of a mixed state,
$\gamma(\rho) \le 2  f_\alpha (C(\rho))$. On the other hand,  in Eqs. (\ref{Eq:f1}) (\ref{Eq:f2}), we have the other form of the function
$f_\alpha(C(\psi_i) = \sqrt{1+ A C^2 (\psi_i)}$ being a monotonously increasing convex function.
Using the property of the monotonously increasing convex function, we have the inequality
\begin{align}
\min_{p_i, \psi_i} \sum_i p_i f_\alpha(C(\psi_i)) \geqslant   f_\alpha(  \min_{p_i, \psi_i} \sum_i p_i C(\psi_i))    = f_\alpha (C(\rho)).
\label{Eq:ineq}
\end{align}

It was shown in Ref. \cite{Wootters:1997id} that there exists a decomposition of $\{p_i, \psi_i \}$ such that all the pure states $\psi_i$
have the same concurrence $C(\psi_i)$. This indicates the concurrence of a mixed state
is exactly the same as the concurrence of pure states in this particular configuration,  $C(\rho) =  \min_{p_i, \psi_i} \sum_i p_i C(\psi_i) = C(\psi_i)$.
Thus, the equality holds in Eq. (\ref{Eq:ineq}),
$\min_{p_i, \psi_i} \sum_i p_i f_\alpha(C(\psi_i)) = f_\alpha (C(\rho)).$
Hence, we can conclude that the maximal violation of the Bell's inequality of a mixed state
is bounded by the function of concurrence of a mixed state
\begin{align}
\gamma(\rho)  \le 2 f_\alpha(C({\rho}))
\label{Eq:T1mixed}
\end{align}.

\section{Applications of Theorem 2.1}
\label{3}
We apply our results [Eqs. (\ref{Eq:T1}) and (\ref{Eq:T1mixed})] to two examples.
The first one is an XY-model with a non-uniform magnetic field. As a demonstration, we only consider the two-qubit case at zero and finite temperature.
The former one corresponds to a pure state and the later corresponds to a mixed state. 
The second example is the Wen-Plaquette model. In the case with only four sites, the ground state has the form as Eq. (\ref{Eq:F1})
and can be applied from Theorem 2.1. On the other hand, in the case with six sites, the ground state is not the form as Eq. (\ref{Eq:F1})
and Theorem 2.1 is not applicable. However, we can still relate the maximal violation of the Bell's inequality and
the {\it generalized} concurrence in this case. We then use this relation to extract the topological entanglement entropy $S_{\rm TEE} = \ln 2$ 
in the Wen-Plaquette model.

\subsection{Application in an XY-model}
\subsubsection{Zero Temperature}
We consider the 2-qubit XY-model with the non-uniform magnetic field
\bea
H_{2\mathrm{qubits}}
&=&-\frac{J}{2}(1+\tilde{\gamma})\sigma_x\otimes\sigma_x-\frac{J}{2}(1-\tilde{\gamma})\sigma_y\otimes\sigma_y
\nn\\
&&-B(1+\delta)\sigma_z\otimes I-B(1-\delta)I\otimes\sigma_z,
\eea
 where 
 \bea
 0\le\tilde{\gamma}\le1, \qquad 0\le\delta\le1, \qquad J\ge 0, \qquad B\ge 0.
 \eea
Thus, the Hamiltonian is:
\bea
H_{2\mathrm{qubits}}
 &=&\begin{pmatrix} -2B& 0 &0&-J\tilde{\gamma}
 \\ 0&2B\delta& -J&0
 \\0&-J&-2B\delta&0
 \\-J\tilde{\gamma} &0&0&2B
 \end{pmatrix}.
\eea
The corresponding eigenvectors and eigenvalues are
\bea
 |\Psi_{1,\pm} \rangle&=&\begin{pmatrix} \mp \sqrt{\frac{|\lambda_{1,\pm}| \mp 2B}{2|\lambda_{1,\pm}|}}
 \\ 0
 \\0
 \\ \sqrt{\frac{|\lambda_{1,\pm}| \pm 2B}{2|\lambda_{1,\pm}|}}
 \end{pmatrix}
 = \mp \sqrt{\frac{|\lambda_{1,\pm}| \mp 2B}{2|\lambda_{1,\pm}|}}|00 \rangle 
+ \sqrt{\frac{|\lambda_{1,\pm}|\pm 2B}{2|\lambda_{1,\pm}|}}|11 \rangle,  \nn\\
| \Psi_{2,\pm} \rangle&=&\begin{pmatrix} 0
 \\ \mp \sqrt{\frac{|\lambda_{2,\pm}|\pm 2B\delta}{2|\lambda_{2,\pm}|}}
 \\ \sqrt{\frac{|\lambda_{2,\pm}|\mp 2B\delta}{2|\lambda_{2,\pm}|}}
 \\ 0
 \end{pmatrix}
 =\mp \sqrt{\frac{|\lambda_{2,\pm}|\pm2B\delta}{2|\lambda_{2,\pm}|}} |10\rangle +  \sqrt{\frac{|\lambda_{2,\pm}|\mp 2B\delta}{2|\lambda_{2,\pm}|}} |01 \rangle,
 \nn\\
\eea
 where  $\lambda_{1,\pm} = \pm\sqrt{4B^2+J^2\tilde{\gamma}^2}$, $\lambda_{2,\pm}=\pm\sqrt{J^2+4B^2\delta^2}$, and
$|00\rangle = (1,0,0,0)^{\rm T}$,  $|11\rangle = (0,0,0,1)^{\rm T}$,  $ |10\rangle =(0,1,0,0)^{\rm T}$, and $ |01\rangle =(0,0,1,0)^{\rm T}$.

We first compute the R-matrices from Eq. (\ref{Eq:RM}) for these states which are all diagonal, 
$R(\Psi_{1,\pm}) = {\rm diag} (\mp  \sqrt{\frac{\lambda_{1,\pm}^2-4B^2}{\lambda_{1,\pm}^2}}, \pm   \sqrt{\frac{\lambda_{1,\pm}^2-4B^2}{\lambda_{1,\pm}^2}} , 1)$ 
and $R(\Psi_{2,\pm}) = {\rm diag} (\mp  \sqrt{\frac{\lambda_{2,\pm}^2-4(B\delta)^2}{\lambda_{2,\pm}^2}}, \mp  \sqrt{\frac{\lambda_{2,\pm}^2-4(B\delta)^2}{\lambda_{2,\pm}^2}} , -1)$.
The maximal violation of the Bell's inequality can be computed from the eigenvalues of $R^\dagger R$ and is given in Eq. (\ref{Eq7}),
$\gamma(\Psi_{1,\pm})=2\sqrt{1+ \frac{\lambda_{1,\pm}^2-4B^2}{\lambda_{1,\pm}^2}}$ 
and $ \gamma(\Psi_{2,\pm})=2\sqrt{1+ \frac{\lambda_{2,\pm}^2-4(B\delta)^2}{\lambda_{2,\pm}^2}}$.

We can compute the concurrence direction from the pure states by Eq. (\ref{Eq:C}),
$C(\Psi_{1,\pm}) = \sqrt{ \frac{\lambda_{1,\pm}^2-4B^2}{\lambda_{1,\pm}^2}}$ and  
$C(\Psi_{2,\pm}) = \sqrt{ \frac{\lambda_{2,\pm}^2-4(B \delta)^2}{\lambda_{2,\pm}^2}}$.
We can immediately see that Theorem 2.1 is held, $\gamma(\Psi_{i,\pm}) = 2 \sqrt{1+C(\Psi_{i,\pm}) ^2} $, $i=1,2$.
This relation demonstrates that we can extract the entanglement in terms of concurrence directly from
the maximal violation of the Bell's inequality. 

\subsubsection{Finite Temperature}
To compute the density matrix at finite temperature, we first compute the following:
\bea
\rho_{AB}
&=&\exp\bigg(-\frac{1}{T} H_{2\mathrm{qubits}}\bigg)\sum_{i=1,2; \alpha=\pm}|\Psi_{i,\alpha}\rangle\langle\Psi_{i,\alpha}|\nn\\
&=&\Bigg(\cosh\bigg(\frac{|\lambda_{1,\pm}|}{T}\bigg)+2\frac{B}{|\lambda_{1,\pm}|}\sinh\bigg(\frac{|\lambda_{1,\pm}|}{T}\bigg)\Bigg)
|00\rangle\langle 00|  \nn\\
&&+\Bigg(\cosh\bigg(\frac{|\lambda_{1,\pm}|}{T}\bigg)-2\frac{B}{|\lambda_{1,\pm}|}\sinh\bigg(\frac{|\lambda_{1,\pm}|}{T}\bigg)\Bigg)
|11\rangle\langle 11|\nn\\
&&+\sinh\bigg(\frac{|\lambda_{1,\pm}|}{T}\bigg)\sqrt{1-\frac{4B^2}{\lambda_{1,\pm}^2}}
\Bigg(|00\rangle\langle 11|+ |11\rangle\langle 00|\Bigg)
\nn\\
&&+\Bigg(\cosh\bigg(\frac{|\lambda_{2,\pm}|}{T}\bigg)-2\frac{B \delta}{|\lambda_{2,\pm}|}\sinh\bigg(\frac{|\lambda_{2,\pm}|}{T}\bigg)\Bigg)
|10\rangle\langle 10|  \nn\\
\nn\\
&&+\Bigg(\cosh\bigg(\frac{|\lambda_{2,\pm}|}{T}\bigg)+2\frac{B \delta}{|\lambda_{2,\pm}|}\sinh\bigg(\frac{|\lambda_{2,\pm}|}{T}\bigg)\Bigg)
|01\rangle\langle 01| \nn\\
&&+\sinh\bigg(\frac{|\lambda_{2,\pm}|}{T}\bigg)
\sqrt{1-\frac{4B^2 \delta^2}{\lambda_{2,\pm}^2}}\Bigg(|10\rangle\langle 01|+ |01\rangle\langle 10|\Bigg),
\eea
where $T$ is a temperature. We can use the density matrix at a finite temperature to compute the concurrence of the mixed state. 
We can determine a critical temperature $T_c$ such that the concurrence of the mixed state vanishes.

The concurrence defined for a mixed state is given in Eq. (\ref{Eq:C_mixed}).
It is shown in Ref. \cite{Wootters:1997id} that Eq. (\ref{Eq:C_mixed}) can be determined by
$C(\rho_{AB})= {\rm max} (0,\xi_1-\xi_2-\xi_3-\xi_4)$, where
$\xi_i$ are the eigenvalues,  in decreasing order, of the quantity $\sqrt{\rho_{AB}(\sigma_y\otimes\sigma_y)\rho_{AB}^*(\sigma_y\otimes\sigma_y)}$:
\bea
\bigg\{ \xi_i \bigg\}
=\bigg\{     
&&\sqrt{\cosh^2 \bigg(\frac{|\lambda_{1,\pm}|}{T}\bigg)-\frac{4B^2}{|\lambda_{1,\pm}|^2}\sinh ^2 \bigg(\frac{|\lambda_{1,\pm}|}{T}\bigg)},
\pm
\sinh \bigg(\frac{|\lambda_{1,\pm}|}{T}\bigg) \sqrt{1-\frac{4B^2}{|\lambda_{1,\pm}|^2}},   \nn \\   
&&\sqrt{\cosh^2 \bigg(\frac{|\lambda_{2,\pm}|}{T}\bigg)-\frac{4B^2 \delta^2}{|\lambda_{2,\pm}|^2}\sinh ^2 \bigg(\frac{|\lambda_{2,\pm}|}{T}\bigg)},
\pm
\sinh \bigg(\frac{|\lambda_{2,\pm}|}{T}\bigg) \sqrt{1-\frac{4B^2 \delta^2}{|\lambda_{2,\pm}|^2}} 
\bigg\}. \nn \\
\eea

The concurrence of the mixed state is
\bea
\max\Bigg( 2\sinh \bigg(\frac{|\lambda_{1,\pm}|}{T}\bigg) \sqrt{1-\frac{4B^2}{|\lambda_{1,\pm}|^2}}
-2\sqrt{\cosh^2 \bigg(\frac{|\lambda_{2,\pm}|}{T}\bigg)-\frac{4B^2 \delta^2}{|\lambda_{2,\pm}|^2}\sinh ^2 \bigg(\frac{|\lambda_{2,\pm}|}{T}\bigg)} ,0
\Bigg)
\nn\\
\eea
when
\bea
&&\sqrt{\cosh^2 \bigg(\frac{|\lambda_{1,\pm}|}{T}\bigg)-\frac{4B^2}{|\lambda_{1,\pm}|^2}\sinh ^2 \bigg(\frac{|\lambda_{1,\pm}|}{T}\bigg)}
+
\sinh \bigg(\frac{|\lambda_{1,\pm}|}{T}\bigg) \sqrt{1-\frac{4B^2}{|\lambda_{1,\pm}|^2}}
 \nn\\
 \geq&&
\sqrt{\cosh^2 \bigg(\frac{|\lambda_{2,\pm}|}{T}\bigg)-\frac{4B^2 \delta^2}{|\lambda_{2,\pm}|^2}\sinh ^2 \bigg(\frac{|\lambda_{2,\pm}|}{T}\bigg)}
+
\sinh \bigg(\frac{|\lambda_{2,\pm}|}{T}\bigg) \sqrt{1-\frac{4B^2 \delta^2}{|\lambda_{2,\pm}|^2}} 
\bigg\},\nn\\
\eea
and the concurrence of the mixed state is
\bea
\max\Bigg( 2\sinh \bigg(\frac{|\lambda_{2,\pm}|}{T}\bigg) \sqrt{1-\frac{4B^2 \delta^2}{|\lambda_{2,\pm}|^2}} 
-2\sqrt{\cosh^2 \bigg(\frac{|\lambda_{1,\pm}|}{T}\bigg)-\frac{4B^2}{|\lambda_{1,\pm}|^2}\sinh ^2 \bigg(\frac{|\lambda_{1,\pm}|}{T}\bigg)} ,0
\Bigg)
\nn\\
\label{Eq:57}
\eea
when
\bea
&&\sqrt{\cosh^2 \bigg(\frac{|\lambda_{2,\pm}|}{T}\bigg)-\frac{4B^2 \delta^2}{|\lambda_{2,\pm}|^2}\sinh ^2 \bigg(\frac{|\lambda_{2,\pm}|}{T}\bigg)}+
\sinh \bigg(\frac{|\lambda_{2,\pm}|}{T}\bigg) \sqrt{1-\frac{4B^2 \delta^2}{|\lambda_{2,\pm}|^2}} 
 \nn\\
\ge&&
\sqrt{\cosh^2 \bigg(\frac{|\lambda_{1,\pm}|}{T}\bigg)-\frac{4B^2}{|\lambda_{1,\pm}|^2}\sinh ^2 \bigg(\frac{|\lambda_{1,\pm}|}{T}\bigg)}
+
\sinh \bigg(\frac{|\lambda_{1,\pm}|}{T}\bigg) \sqrt{1-\frac{4B^2}{|\lambda_{1,\pm}|^2}}.\nn\\
\label{Eq:58}
\eea

First, we consider when $B=0$ which $|\lambda_1| = J \tilde{\gamma} < J = |\lambda_2|$
and Eq. (\ref{Eq:58}) hold. Hence the concurrence is given by Eq. (\ref{Eq:57})
\bea
C(\rho_{AB})=\max\Bigg(2 \sinh \bigg( \frac{ |\lambda_{2,\pm}|}{T}\bigg)-2 \cosh   \bigg( \frac{ |\lambda_{1,\pm}|}{T}\bigg) ,0\Bigg),  \nn\\
\eea
The critical temperature $T_c$ can be found  from the equation
\bea  
&&\sinh \bigg( \frac{ J}{T_c}\bigg)=    \cosh   \bigg( \frac{ J \gamma }{T_c}\bigg).
\eea

Next, we consider $\tilde{\gamma}=0$ and $\delta=1$ case.
 In this case, $|\lambda_1|  < |\lambda_2|$  and the concurrence of the mixed state is given by Eq. (\ref{Eq:57})
\bea
C(\rho_{AB})
=\max\Bigg( 2\sinh \bigg(\frac{\sqrt{J^2+4B^2}}{T}\bigg) \sqrt{1-\frac{4B^2 }{J^2+4B^2}}
-2 ,0
\Bigg).
\eea
The critical temperature is
\bea
T_c=\frac{\sqrt{J^2+4B^2 }}{\sinh^{-1}\bigg(\sqrt{1+\frac{4B^2 }{J^2}}\bigg)}.
\eea

We observe that the critical temperature depends on the magnetic field $B$.  
When we take the limit $J\rightarrow\infty$, the critical temperature $T_c$ approaches to infinity. 
This is expected because the thermal state with infinite $J$ is still highly entangled.
 When we take another limit $J\rightarrow 0$, the concurrence of the mixed state approaches zero, 
 since all spins are aligned along the direction of the magnetic field and will not be entangled.
  Thus, the mixed state approaches to a product state for each temperature in the limit.
 One interesting limit is taking the limit $B/J\rightarrow\infty$.  The critical temperature approaches to $J$.
 These limits provide useful and interesting applications to entanglement from the concurrence of the mixed state at a finite temperature. 
 
 According to Eq. (\ref{Eq:T1mixed}), if one measures the extremal Bell's violation, it will be bounded by the function of the concurrence of the mixed states,
\bea
\max_{{\cal B}_n}\mbox{Tr}(\rho{\cal B}_n)\le2f_2(\rho_{AB})=2\sqrt{1+C^2(\rho_{AB})}.
\label{Eq:63}
\eea

\subsection{Applications to the Wen-Plaquette Model}
Next example we considered is the Wen-Plaquette model, which is a two-dimensional spin (qubit) model with 
the following interactions,
\bea
H_{\mathrm{WP}}=\sum_{i} \sigma_x^i\sigma_y^{i+\hat{x}}\sigma_x^{i+\hat{x}+\hat{y}}\sigma_y^{i+\hat{y}}.
\eea
Here qubits live on vertices with a four-spin interaction on each plaquette.
A quantum state of the Hamiltonian is an $n$-qubit quantum state with $n$ being a number of vertices.
First, we apply our Theorem 2.1 to a four-qubit quantum state, with the geometry of the system containing four vertices, 
eight edges, and four faces, in which the Euler number of the torus is zero [see Fig. \ref{F1}(a)],
$\chi= V-E+F=0$
 with $V$, $E$, and $F$ being a number of vertices, edges, and faces, respectively.
There are four degenerate ground states $|G\rangle_{\rm 4-qubit}$: 
\bea
 &&\frac{1}{\sqrt{2}} (|0000\rangle + |1111 \rangle), 
\quad
\frac{1}{\sqrt{2}} (|1010\rangle + |0101 \rangle),  
 \nn\\
 &&\frac{1}{\sqrt{2}} (|0011\rangle - |1100 \rangle), 
\quad
 \frac{1}{\sqrt{2}} (|1001\rangle - |0110 \rangle).
\eea
We define the order of each site in these four-qubit states in the Fig. \ref{F1} (a). 
We first compute the R-matrices from Eq. (\ref{Eq:RM}) for these states which are $27 \times 3$ matrices.
The $R^\dagger R$ of different states have same diagonal form, $R^\dagger R = \rm{diag} (4,4,1)$, which lead to 
the maximal violation of Bell's inequality $\gamma =4\sqrt{2}$ from the Lemma 2.1 (Eq. [\ref{Eq7})].

Next we compute the concurrence directly from these ground state, $C(|G\rangle_{4-\mathrm{qubit}})=1$.
We immediately see that
Theorem 2.1 [ Eq. (\ref{Eq:T1})] still holds. 
The maximal violation of the Bell's inequality 
is $\gamma=2\times 2^{\frac{4-1}{2}}=4\sqrt{2}$.
Here we demonstrate that  the maximal violation of the Bell's inequality can be directly computed from the concurrence
in the 4-qubit Wen-Plaquette model.
\begin{figure}
\begin{centering}
\includegraphics[scale=0.35]{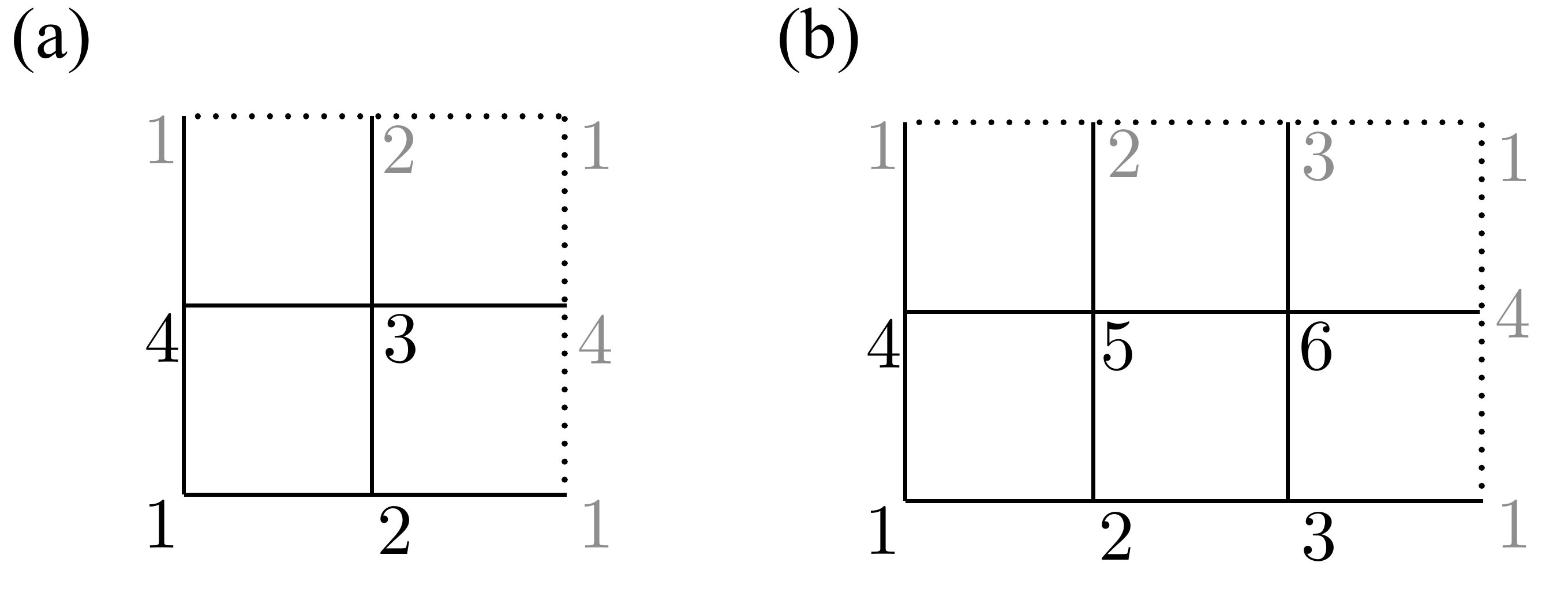}
\par\end{centering}
\caption{(a) A four-qubit quantum state in the Wen-Plaquette model and (b) A six-qubit quantum state of the Wen-Plaquette model on a torus.
  The right dashed line is identified as the left solid line and the top dashed line is identified as the bottom solid line in (a) and (b).
  The numbers are the site indices. The gray colored number is identified with the corresponding black colored number.}
  \label{F1}
\end{figure}

Next, we would like to consider a six-qubit states. The ground state of the Wen-Plaquette model 
are
\bea
|G_1\rangle_{\rm 6-qubit}
&=&\frac{1}{2}\big(-|111000\rangle +|001110\rangle+|100011\rangle+|010101\rangle\big),  
\eea
\bea
|G_2\rangle_{\rm 6-qubit}
=\frac{1}{2}\big(-|000111\rangle +|110001\rangle + |011100\rangle+|101010\rangle\big),  
\eea
with the site labels shown in Fig. \ref{F1} (b). 
Obvious these states do not have the form of Theorem 2.1 and thus 
Theorem 2.1 is not applicable.

To relate the upper bound of the maximal violation of the Bell's inequality and the {\it generalized} concurrence (which we will define later),
we need to rewrite the ground states in terms of two parameters $(\lambda_+,\lambda_-)$. Without loss of generality, we choose
\bea
\label{EQ 123}
&&|\psi_1\rangle_{\rm 6-qubit}
=\frac{\lambda_+}{\sqrt{2}}\big(-|111000\rangle +|001110\rangle \big)
+\frac{\lambda_-}{\sqrt{2}}\big(|100011\rangle+|010101\rangle\big),  
\nn\\
&&|\psi_2\rangle_{\rm 6-qubit}
=\frac{\lambda_+}{\sqrt{2}}\big(-|000111\rangle +|110001\rangle \big)
+\frac{\lambda_-}{\sqrt{2}}\big(|011100\rangle+|101010\rangle\big).
\eea
When $\lambda_+=\lambda_-=\frac{1}{\sqrt{2}}$, the above states are the ground states of the six-site Wen-Plaquette model.

Now we can relate the upper bound of the maximal violation of the Bell's inequality to the generalized concurrence of the pure states given in Eq. (\ref{EQ 123}).
Two different bipartitions are considered: (1) subsystem $A$ contains a site number six, and (2) subsystem $A$ contains a site number five and a site number six.
Here we use $\delta=1$ or $2$ as an indicator for the case one and case two.
According to the Lemma, we find that the upper bound of the maximal violation of the Bell's inequality can be expressed by
the eigenvalues of the $R^\dagger R$ and will be a function of $(\lambda_+, \lambda_-)$.
We then find the inverse mapping of the generalized concurrence such that these two parameters can
be expressed in terms of the generalized concurrence ($C$),  $(\lambda_+(C), \lambda_-(C))$.
Thus, the maximal Bell's violation can be written as a function of the generalized concurrence $\gamma(C)$.
 
One should notice that the  maximal Bell's violation is bipartition independent but measures of entanglement are not,
i.e., the concurrence, the entanglement entropy, or the purity may depend on how we bipartite the system.
Here we define the generalized concurrence as following,
\bea
C(\delta)\equiv\sqrt{2\bigg(1-2^{\delta-1}\mbox{Tr}\rho_{A(\delta)}^2\bigg)},
\label{EQ:GC}
\eea 
which $\delta$ indicates different bipartitions.
In \ref{app1}, we show that the above definition keeps the relation of the concurrence and entanglement entropy be independent on the bipartition $\delta$ in
the toric code model. However, the purity ${\rm Tr} \rho_{A(\delta)}^2$ still depends on the bipartition.
We find that the relation of the generalized concurrence and entanglement entropy is also independent on the bipartitions in the 6-site Wen-Plaquette model as well. Hence the maximal violation of the 
Bell's inequality will be the same functional form of the generalized concurrence as we will demonstrate in the follows.  
 
Without repeating the similar calculation, 
we only give the computation of the quantum state ($|G_1\rangle_{\rm{6-qubit}}$). 
To obtain the upper bound of the maximal violation of the Bell's inequality, we first compute the density matrix
\bea
&&\rho_{{\rm 6-qubit}, 1}
\nn\\
&=&\frac{\lambda_+^2}{2}\big(|111000\rangle\langle 111000| +|001110\rangle\langle 001110|
-|111000\rangle\langle 001110|-|001110\rangle\langle111000|\big)
\nn\\
&&+\frac{\lambda_-^2}{2}\big(|100011\rangle\langle 100011| +|010101\rangle\langle 010101|
+|100011\rangle\langle 010101|+|010101\rangle\langle 100011|\big)
\nn\\
&&+\frac{\lambda_+\lambda_-}{2}
\big(-|111000\rangle\langle 100011|-|111000\rangle\langle 010101|
+|001110\rangle\langle 100011|+|001110\rangle\langle 010101|
\nn\\
&&-|100011\rangle\langle 111000|-|010101\rangle\langle 111000|
+|100011\rangle\langle 001110|+|010101\rangle\langle 001110|\big).
\eea

The non-vanishing matrix elements of the generalized $R$-matrix are:
\bea
&&R_{zzzzzz}=R_{yyzxxz}=R_{xxzyyz}=-1,
\nn\\
&& R_{xxzxxz}=R_{yxzyxz}=R_{xyzxyz}=R_{yyzyyz}=\lambda_+^2-\lambda_-^2
 \nn\\
 &&R_{yxzxyz}=R_{xyzyxz}=1,
\nn\\
&&R_{zyyzxx}=-2\lambda_+\lambda_-,  \qquad R_{zxyzyx}=2\lambda_+\lambda_-
\nn\\
&&R_{yzyxzx}=-2\lambda_+\lambda_-, \qquad R_{xzyyzx}=2\lambda_+\lambda_-
\nn\\
&&R_{zyxzxy}=2\lambda_+\lambda_-, \qquad R_{zxxzyy}=-2\lambda_+\lambda_-, 
\nn\\
&&R_{yzxxzy}=2\lambda_+\lambda_-, \qquad R_{xzxyzy}=-2\lambda_+\lambda_-.
\label{Eq71}
\eea
The eigenvalues of the matrix $R^{\dagger}R$ are $(5+4(\lambda_+^2-\lambda_-^2)^2,  16\lambda_+^2\lambda_-^2,  16\lambda_+^2\lambda_-^2)$.

Now we want to compute the concurrence of these states.
If we denote the last qubit as the region $A$ and the complementary region as the region $B$, 
the concurrence is $C(\lambda_+, \lambda_-) = 2 |\lambda_+\lambda_-|$.
We also have the normalization condition of the state $1= \lambda_+^2 +\lambda_-^2$. 
Therefore, we can invert the function of the concurrence and have
\bea
\lambda_+^2=\frac{1\pm\sqrt{1-C^2(1, \psi)}}{2},\qquad 
 \lambda_-^2=\frac{1\mp\sqrt{1-C^2(1, \psi)}}{2}.
 \label{Eq72}
\eea
Now the eigenvalues of the matrix $R^{\dagger}R$  can be expressed as 
$(9-4C^2(1, \psi),  4C^2(1, \psi),  4C^2(1, \psi))$.
The upper bound of maximal violation of the Bell's inequality from the eigenvalues of the matrix $R^{\dagger}R$ [Eq. (\ref{Eq7})] is 
a function of the concurrence
\begin{align}
\gamma (C) \le
\begin{cases}
6,  \quad   &C^2(1, \psi)\le\frac{9}{8}  \\
2 \sqrt{8C^2(1, \psi)},  \quad    &C^2(1, \psi)>\frac{9}{8}.
\end{cases}
\end{align}

Next, we consider difference bipartition of computing the generalized concurrence defined in Eq. (\ref{EQ:GC}). 
If we denote the last two qubits as the region $A$ and the complementary region as the region $B$, 
the generalized concurrence is $C(2,\psi) = 2 |\lambda_+\lambda_-|$ and its inverse mapping gives the same form as Eq. (\ref{Eq72}).
The above result shows that the upper bound of the maximal violation of Bell's inequality has the same functional form for the generalized concurrence
with different bipartitions. 

One should be noticed that the generalized concurrence we defined is always less or equal then one, $C(\delta, \psi) \le 1$.
This indicates the upper bound of the maximal violation of the Bell's inequality is aways six which is independent on the states and bipartitions.
This independency is purely an artifact from the choice of the Bell's operator [Eq. (\ref{Eq:mBell})].
We can find another Bell's operator to have a better upper bound.
This choice is switching site-1 and site-6 from the original Bell's operator defined in Eq. (\ref{Eq:mBell}). 
We recompute the upper bound of the maximal violation of the Bell's inequality for this choice of the Bell's operator.
The non-vanishing elements of the $R$-matrix in this Bell's operator are switching the first and six indices of the  $R$-matrix
 in Eq. (\ref{Eq71}). The corresponding eigenvalues of the matrix $R^\dagger R$ are $(4, 4, 1+16 \lambda_+^2\lambda_-^2)$.
 The generalized concurrence under redefinition of the site-1 and site-6 are the same as the original one. Hence the  
 eigenvalues of the matrix $R^\dagger R$ are $(4, 4, 1+4 C^2(\delta, \psi))$, $\delta=1,2$.
 
 The upper bound of the maximal violation of the Bell's inequality as a function of the generalized concurrence is
 \begin{align}
\gamma (C) \le
\begin{cases}
2\sqrt{5+4 C^2(\delta,\psi)},  \quad  &C^2(\delta,\psi)\ge\frac{3}{4}  \\
4 \sqrt{2},  \quad   &C^2(\delta,\psi)<\frac{3}{4}, \quad \delta=1,2.
\end{cases}
\end{align}

One can find the ground state of the six-site Wen-Plaquette model ($\lambda_+=\lambda_- = 1/\sqrt{2}$)
has the upper bound of the Bell's inequality equal to $6$ and the concurrence $C=1$. 

Here we give a brief comment on the choice of Bell's operators.
The essential idea of using the upper bound of the maximal Bell's violation as a measure of entanglement is to express the upper bound of the maximal Bell's violation as a function of the concurrence. 
The functional form of the upper bound of the maximal Bell's violation strongly depends on the choice of the Bell's operators. 
If it is a constant for a given region of the concurrence, e.g. $C^2\ge  3/4$ for the first choice of the Bell's operator
and $C^2<3/4$ for the second choice of the Bell's operator, we cannot extract the entanglement information form the upper bound of the maximal Bell's inequality.
This is different from Theorem 2.1 where the maximal violation of the Bell's inequality is always a function of the concurrence
with the particular form of the wavefunction.

Now let us suppose that we choose the first Bell's operator and we know that the generalized concurrence is always greater than $\sqrt{3}/2$.
Can we extract the topological entanglement entropy from measuring the upper bound of the maximal violation of the Bell's inequality?
The answer is yes. First we write down the entanglement entropy $S_{\mathrm{EE}, A(\delta)}=-{\rm Tr} \rho_{A(\delta)} \ln \rho_{A(\delta)} $ for different bipartitions,
 \bea
 S_{\mathrm{EE}, A(1)}=-\lambda_+^2\ln\lambda_+^2-\lambda_-^2\ln\lambda_-^2, 
 \qquad 
 S_{\mathrm{EE}, A(2)}=\ln 2-\lambda_+^2\ln\lambda_+^2-\lambda_-^2\ln\lambda_-^2.
 \nn\\
 \eea

Next we express $\lambda_\pm^2$ in term of the upper bound of the maximal violation of the Bell's inequality $\gamma_M$,  
\bea
\lambda_\pm^2 = \frac{1}{2} \bigg(1\pm\sqrt{\frac{9}{4}-\frac{\gamma_M^2}{16}}\bigg).
\eea
Thus the entanglement entropy can be written as the function fo the concurrence,
\bea
 S_{\mathrm{EE}, A(1)} &=& -\frac{1}{2} \ln \bigg( \frac{\gamma_M^2}{4^3}-\frac{5}{4^2}\bigg)-\frac{1}{2}\sqrt{\frac{9}{4}-\frac{\gamma_M^2}{16}} \ln \bigg( \frac{1+\sqrt{\frac{9}{4}-\frac{\gamma_M^2}{16}}}{1-\sqrt{\frac{9}{4}-\frac{\gamma_M^2}{16}}}\bigg), \nn\\
  S_{\mathrm{EE}, A(2)} &=& S_{\mathrm{EE}, A(1)} + \ln 2.
\eea

For the ground state of the six-site Wen-Plaquette model, $\gamma_M=6$ leads to
$ S_{\mathrm{EE}, A(1)} = \ln 2$ and $ S_{\mathrm{EE}, A(2)}= 2 \ln 2$.

In general, the entanglement entropy has the form 
\bea
S_{\mathrm{EE}, A(L)}= \alpha L  -S_{\rm TEE},
\eea
in which the first term indicates the area law with $L$ being the length of an entangling boundary, $\alpha$ being a constant,
and $S_{\rm TEE}$ is called topological entanglement entropy \cite{Kitaev:2005dm}.
In the Wen-Plaquette model, the length of an entangling boundary $L$ is a number of bonds that connect the subsystem $A$ and the subsystem $B$.
We consider the following number of bonds:
\bea
L(\delta=1)=4, \qquad L(\delta=2)=6
\eea
to extract the area law of entanglement entropy and obtain the topological entanglement entropy:
\bea
S_{\rm TEE}=\ln 2=\ln\sqrt{D},
\eea
 where  $D=4$
  is the number of distinct quasiparticles  \cite{Kitaev:2005dm}. In the case that the subsystem $A$ contains one site corresponds to the four bonds $L=4$. In the case, that subsystem $A$ contains two adjacent sites along a vertical direction corresponds to the four bonds $L=4$. In the case that the subsystem $A$ contains two adjacent sites along a horizontal direction corresponds to the six bonds $L=6$. In the case that the subsystem $A$ contains two disjointed sites corresponds to the six bonds $L=6$. In the case that the subsystem $A$ contains three adjacent sites corresponds to the six bonds $L=6$. In the case that the subsystem $A$ contains two adjacent sites and one disjointed site corresponds to the six bonds $L=6$.
Thus, we use the maximal Bell's violation from the generalized $R$-matrix constructing from the first Bell's operator 
to demonstrate an indirect measure of the topological entanglement entropy.

\section{Generalized Concurrence, Entanglement Entropy and $2n$ Qubits}
\label{4}
The last example we would like demonstrate the relation between the upper bound of the maximal Bell's violation
and the generalized concurrence is a $2n$ qubits with the following form
\bea
|\psi\rangle
&=&\frac{\lambda_+}{\sqrt{2^n}}
\bigg(|00\cdots 00_A\rangle|00\cdots 00_B\rangle
+
|00\cdots010_A\rangle|00\cdots010_B\rangle
\nn\\
&&+\cdots
+ |11\cdots110_A\rangle|11\cdots110_B\rangle\bigg)
\nn\\
&&+\frac{\lambda_-}{\sqrt{2^n}} \bigg(|00\cdots 01_A\rangle|00\cdots 01_B\rangle
+
|00\cdots011_A\rangle|00\cdots011_B\rangle
\nn\\
&&+\cdots
+ |11\cdots111_A\rangle|11\cdots111_B\rangle\bigg),
\eea
in which each region $A$ and region $B$ contains $n$ qubits. For each region, apart from the the last qubit, we consider a linear superposition of all possible configurations in the first $n-1$ qubits.
The generalized concurrence of the pure state is defined as before [Eq. (\ref{EQ:GC})] 

The entanglement entropy of the region $A$ is
\bea
S_{\mathrm{EE}, A}
=(n-1)\ln2-\frac{\lambda_+^2}{2}\ln\bigg(\frac{\lambda_+^2}{2}\bigg)
-\frac{\lambda_-^2}{2}\ln\bigg(\frac{\lambda_-^2}{2}\bigg).
\eea
The coefficients, $\lambda_+$ and $\lambda_-$, for $2n$ qubits also satisfy the following equations:
\bea
\frac{\lambda_+^2}{2} = \frac{1\pm\sqrt{1-C(n, \psi)^2}}{2}, 
\quad
\frac{\lambda_-^2}{2}= \frac{1\mp\sqrt{1-C(n, \psi)^2}}{2}.
\eea
The entanglement entropy of the region $A$ also monotonically increases with respect to concurrence of the pure state $C(n, \psi)$.

Now we first compute and discuss the maximal violation of the Bell's inequality for $n=2$. We use the identities given in \ref{app2} to compute
the non-vanishing elements of the generalized $R$-matrix
\bea
&&R_{xxxx}=R_{yyyy}=R_{yxyx}=R_{zxzx}=\lambda_+\lambda_- , \nn\\
&&R_{xyxy}=R_{zyzy}=-\lambda_+\lambda_-, \quad
R_{xzxz}=R_{zzzz}=-R_{yzyz}=1.
\eea
Thus, the eigenvalues of the matrix $R^{\dagger}R$ are 
$(3C(2,\psi)^2,  3C(2,\psi)^2,  3)$.

Now we can obtain the upper bound of the maximal violation of the Bell's inequality
\bea
\gamma\le2\sqrt{3}\sqrt{1+C(2,\psi)^2}.
\eea

We consider different bipartitions, the quantum state can be rewritten as
\bea
|\psi\rangle
&=&\frac{\lambda_+}{2}\big(|000_A\rangle +|101_A\rangle\big)|0_B\rangle
+\frac{\lambda_-}{2}(\alpha_0-\alpha_1)\big(|010_A\rangle +|111_A\rangle\big)|1_B\rangle.
\eea
the entanglement entropy of the region $A$ is
\bea
S_{\mathrm{EE}, A}
=-\frac{\lambda_+^2}{2}\ln\bigg(\frac{\lambda_+^2}{2}\bigg)
-\frac{\lambda_-^2}{2}\ln\bigg(\frac{\lambda_-^2}{2}\bigg),
\eea
and the upper bound of the maximal violation of the Bell's inequality is
\bea
\gamma\le2\sqrt{3}\sqrt{1+C(1, \psi)^2}.
\eea
The entanglement entropy of the region $A$ also monotonically increases with respect to the concurrence of the pure state $C(1, \psi)$.
Although the bipartition does not affect the maximal violation of the Bell's inequality, we can use the concurrence of the pure state to express the intensity of entanglement entropy or how large of entanglement entropy for the corresponding bipartition from the upper bound of the maximal violation of the Bell's inequality.

This example also shows that even if a quantum state is not just a linear combination of two product states, 
the upper bound of the maximal violation of the Bell's inequality can relate to the value of entanglement entropy.

For a generalization of an arbitrary number of $n$, the upper bound of the maximal violation of the Bell's inequality
\bea
\gamma\le2\sqrt{3^{n-1}}\sqrt{1+C(n, \psi)^2}.
\eea

\section{Discussion and Outlook}
\label{5}
We demonstrated the relations between the upper bound of the maximal violation of the Bell's inequality 
and the generalized concurrence of $n$-qubit pure states in several examples.
In particular, we show the this upper is equal to the maximal Bell's violation for an $n$-qubit state which has the form in Eq. (\ref{Eq:F1}).
We emphasize that the maximal violation of the Bell's inequality is a measure of the Bell's operator in a full system,
while the computation of the generalized concurrence only involves a reduced density matrix in a subsystem.
The relation that we showed in the paper can provide an alternative measure of entanglement, from the inverse mapping of the upper bound of the Bell's operator to the generalized concurrence. 

We applied our results to the various models. We first considered the two-qubit system with the non-uniform magnetic field at a finite temperature. 
We also studied the maximal violation of the Bell's inequality in the Wen-Plaquette model \cite{Wen:2003yv} for four sites and six sites on a torus manifold. Our computation of the generalized $R$-matrix in the Wen-Plaquette model reveals that the ground states are maximally entangled or the generalized concurrence of the pure state is one. We also provided a possible detection of topological entanglement entropy \cite{Kitaev:2005dm} through the upper bound of the maximal violation of the Bell's inequality by using different bipartitions of the concurrence of the pure state.

The studies of the upper bound of the maximal violation of the Bell's inequality can shed the light to understand an alternative detection of entanglement beyond using the reduced density matrix. 
Extracting the information of the reduced density matrix form the local probe requires single-site resolution which is very restricted in experiments and is very hard to apply to many-body systems. 
On the other hand, the measure of the Bell's operator, in principle, does not require single-site resolution and can be performed to a large system.
It was shown that Bell's operators have been measured in a Bose-Einstein condensate with about 480 atoms\cite{Schmied:2016}.

\section*{Acknowledgments}
We would like to thank Ling-Yan Hung and Xueda Wen for their insightful discussion. Po-Yao Chang was supported by the Rutgers Center for Materials Theory. Su-Kuan Chu was supported by the AFOSR, NSF QIS, ARL CDQI, ARO MURI, ARO, and NSF PFC at JQI. Chen-Te Ma was supported by the Post-Doctoral International Exchange Program. Chen-Te Ma would like to thank Nan-Peng Ma for his encouragement. We would like to thank the National Tsing Hua University, Tohoku University, Okinawa Institute of Science and Technology Graduate University, Yukawa Institute for Theoretical Physics at the Kyoto University, Istituto Nazionale Di Fisica Nucleare - Sezione di Napoli at the Università degli Studi di Napoli Federico II, Kadanoff Center for Theoretical Physics at the University of Chicago, Stanford Institute for Theoretical Physics at the Stanford University, Kavli Institute for Theoretical Physics at the University of California, and Israel Institute for Advanced Studies at the Hebrew University of Jerusalem. Discussions during the workshops, ``Novel Quantum States in Condensed Matter 2017'', ``The NCTS workshop on correlated quantum many-body systems: from topology to quantum criticality'', ``String-Math 2018'', ``Strings 2018'', ``New Frontiers in String Theory'', ``Strings and Fields 2018'', ``Order from Chaos'' , ``NCTS Annual Theory Meeting 2018: Particles, Cosmology and Strings'', and ``The 36th Jerusalem Winter School in Theoretical Physics - Recent Progress in Quantum Field / String Theory'' were useful to complete this work.

\appendix

\section{Generalized Concurrence, Entanglement Entropy and the Toric Code Model}
\label{app1}
We compute Rényi entropy \cite{Hamma:2005zz} in a toric code model \cite{Kitaev:1997wr} on a disk manifold, which can have an arbitrary number of holes, and a cylinder manifold \cite{Zhang:2011jd} with boundary conditions. From our result of the Rényi entropy, we can obtain a relation between a generalized concurrence of a pure state and entanglement entropy. For the case of a disk manifold, we find that the Rényi entropy always equals to entanglement entropy for any bipartition. For the case of a cylinder manifold, we consider a non-contractible region. The Rényi entropy can be different from the entanglement entropy \cite{Zhang:2011jd}. The result is interesting because we can find that a choice of a "generalized concurrence" depends on boundary degrees of freedom of a Hilbert space in a toric code model. We first review a toric code model on a torus manifold, then compute the Rényi entropy in a toric code model on a disk manifold and on a cylinder manifold. The generalized concurrence is useful for us to demonstrate a relationship between the maximal violation of the Bell's inequality and entanglement entropy in the $2n$-qubit quantum state.

\subsection{Review of a Toric Code Model on a Torus Manifold}
We consider a lattice, which can be embedded on an arbitrary two-dimensional surface. The simplest case is an $L \times L$ square lattice with a periodic boundary condition which forms a torus. The Hilbert space of each edge ${\cal H}_i$ consists of a spin one-half degree of freedom and the total Hilbert space of a toric code model is a tensor product of the Hilbert space of each edge as the Hilbert space ${\cal H}=\bigotimes_i{\cal H}_i$. The lattice model on a torus has $L^2$ vertices and $L^2$ plaquettes. Since the dimensions of the Hilbert space of each spin is two and the number of links on a square lattice is 2$L^2$, the total dimensions of the Hilbert space are $2^{2L^2}$. The Hamiltonian of the toric code model is
\begin{equation}
	\label{eq:hamil}
	H=-\sum_{v}U_vA_v-\sum_{p}J_pB_p \,
\end{equation}
with $U_v\ge 0$ and $J_p \ge0$, where:
\bea
	A_v&=&\bigotimes_{i\in {\mathrm{star}(v)}}\sigma_i^z=\sigma_{i_1}^z\otimes\sigma_{i_2}^z\otimes\sigma_{i_3}^z\otimes\sigma_{i_4}^z \,,
	\nn\\
	B_p&=&\bigotimes_{j\in \partial p}\sigma_j^x=\sigma_{j_1}^x\otimes\sigma_{j_2}^x\otimes\sigma_{j_3}^x\otimes\sigma_{j_4}^x \,,
\eea
in which the index $i\in \mbox{star}(v)$ runs over all four edges around an vertex $v$ and $j\in \partial p$ goes around four edges on a boundary of a plaquette $p$. We also remind that the operators $A_v$ (vertex operators), and $B_p$ (plaquette operators) act on the Hilbert space ${\cal H}$ rather than acting on some local Hilbert spaces so the operators are trivial operators or the identity operators on the edges, which are not in the vertex $v$ and the plaquette $p$.

To have a complete set of observables on a torus, it turns out that we also need two loop operators of two non-contractible cycles on a torus:
\begin{equation}
Z(C_1)\equiv\bigotimes_{i \in C_1} \sigma_{i_1}^z \,, \qquad X(C_2)=\bigotimes_{j \in C_2} \sigma_{j_1}^x \, ,
\end{equation}
in which $C_1$ and $C_2$ are loops. 

\subsection{A Disk Manifold with Holes}
We compute entanglement entropy in the toric code model on a disk manifold with an arbitrary number of holes \footnote{It is also referred to a surface code model if a manifold is not a torus.}.
 A non-zero number of holes on a disk manifold can increase ground state energy of the toric code model compared to ground state energy of the toric code model on the disk manifold without holes. This could be seen as a generation of anyons on a disk manifold through increasing the number of holes. Note that we pick boundary conditions, from which an upper edge of the disk manifold and
a right edge of the disk manifold are smooth boundary conditions and a lower edge of the disk manifold and the left edge of the disk manifold are
rough boundary conditions. A vertex operator at a smooth boundary only acts on three qubits on the three links meeting at a vertex. A plaquette operator at a rough boundary lacks a qubit so it only acts on three qubits on the three links nearby.

When we do a bipartition, there are $(n_{L,v},n_{L,c},n_{R,v},n_{R,c})$ anyons,
in which the label $L$ stands for a left region in the system and the label $R$ stands for a right region in the system and the label $v$ represents vortex particles and the label $c$ represents charged particles. In our convention, the vortex particles are generated by the $X$ operators acting on links and the charge particles are generated by the $Z$ operators acting on a dual lattice. For the convenience, we remove the vertex
operators and plaquette operators corresponding to a position of
the anyons so that the quantum state is, in fact, a ground state of the Hamiltonian \footnote{Now the quantum states with the presence and the absence of anyons are all ground states of the Hamiltonian, we have a larger ground state degeneracy than a toric code model on a disk manifold without any hole.}. We also find that entanglement entropy of the disk manifold with a number of holes is independent of cutting because we can use vertex or the plaquette operators to deform the $X$ and $Z$ operators, as shown in Fig. \ref{fig:movecharge} and Fig. \ref{fig:movevortex}. The other decomposition method can be implemented by putting boundary conditions on the entangling surface \cite{Hung:2015fla}. We can set the rough boundary condition on the left side of an entangling surface and set the smooth boundary condition on the right side of the entangling surface. A form of the entanglement entropy is not modified.

\begin{figure}
\begin{centering}
\includegraphics[scale=0.2]{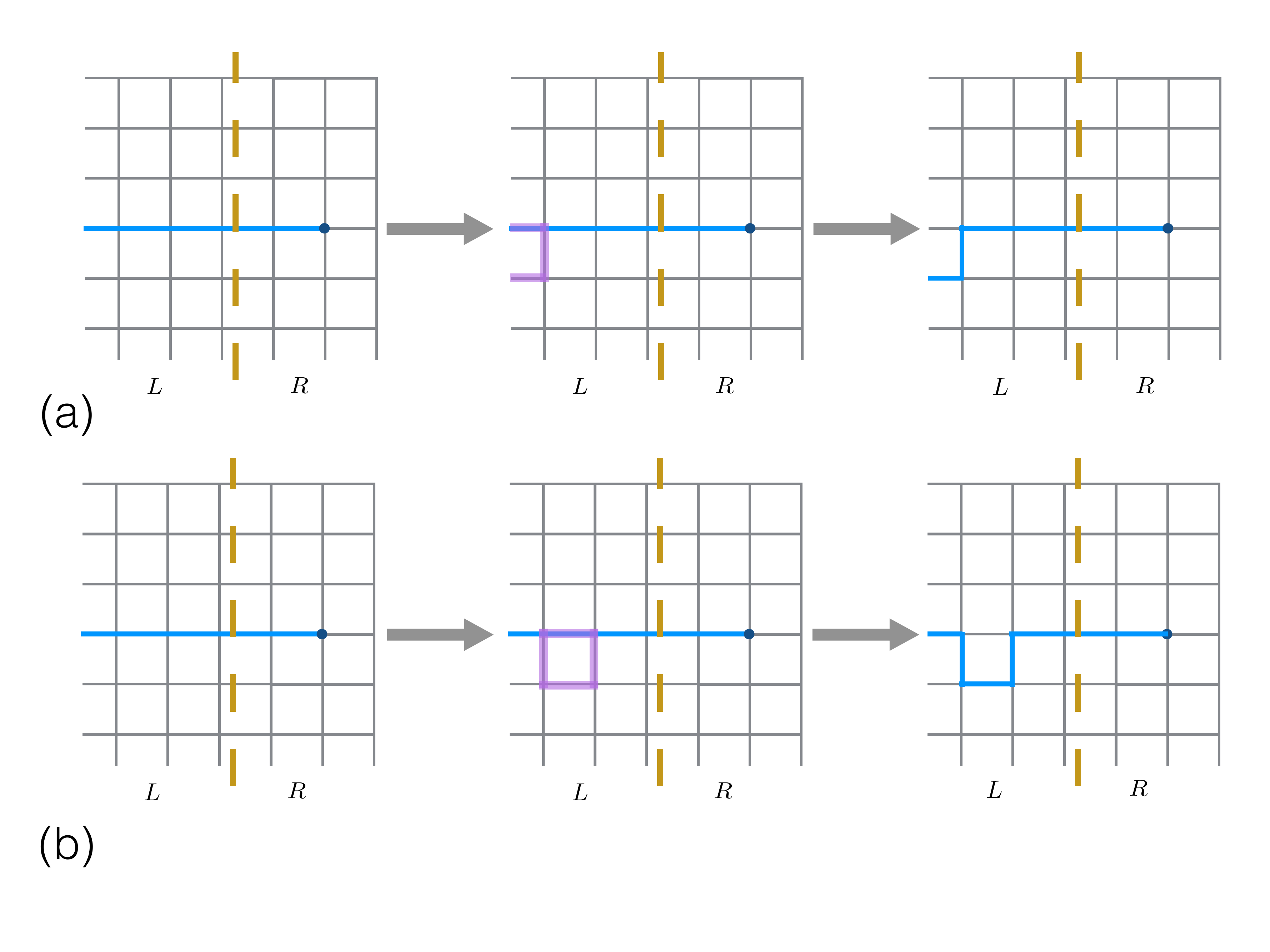}
\par\end{centering}
\caption{We can move the line operator creating charge anyons by applying plaquette operators on the edge as shown in Figure (a) or in the bulk as shown in Figure (b). }
\label{fig:movecharge}
\end{figure}

\begin{figure}
\begin{centering}
\includegraphics[scale=0.2]{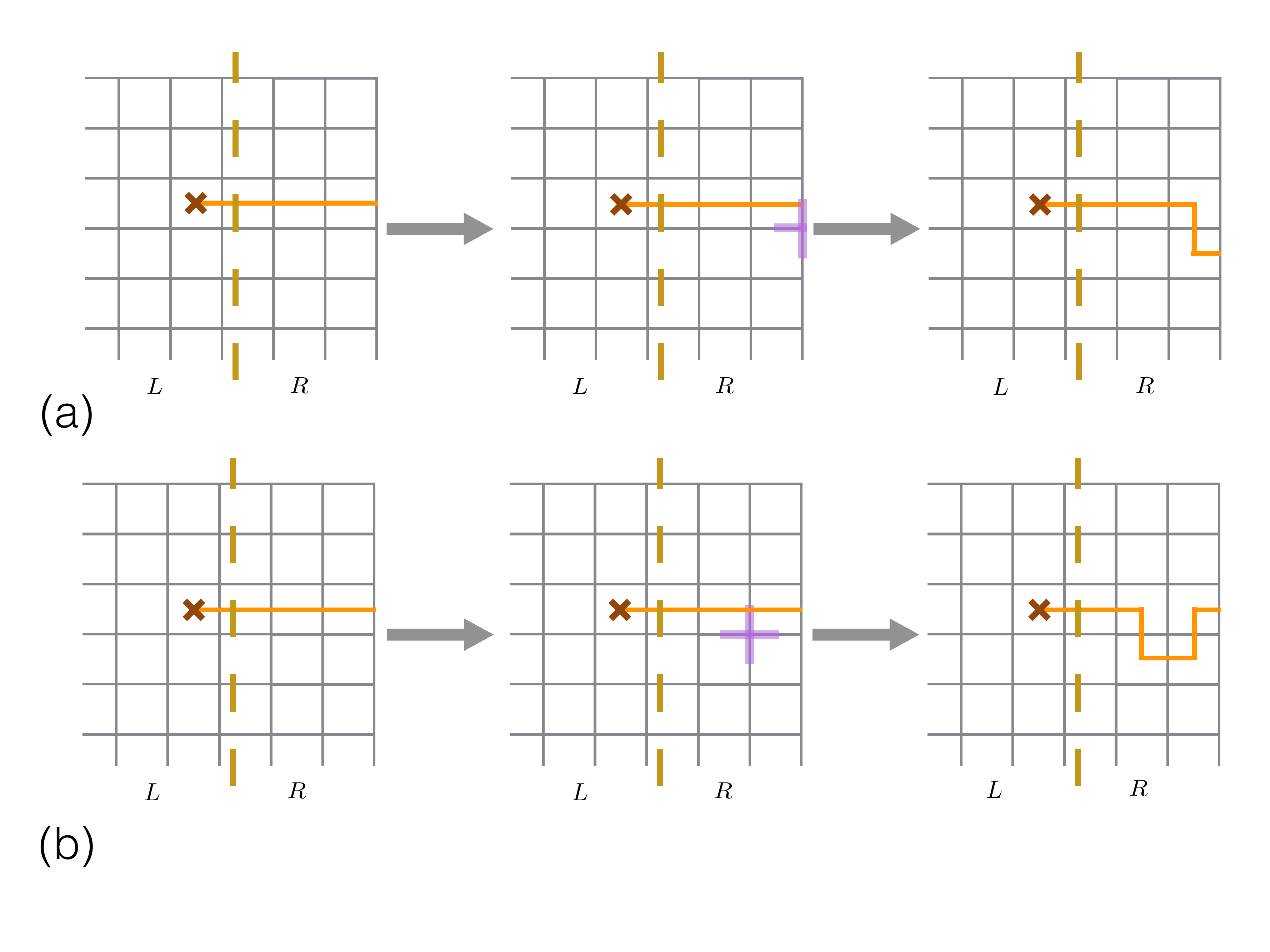}
\par\end{centering}
\caption{We can move the line operator creating vortex anyons by applying vertex operators on the edge as shown in Figure (a) or in the bulk as shown in Figure (b). }
\label{fig:movevortex}
\end{figure}

We define a group $G$ as the group generated by all the plaquette
operators, including those sitting at an edge and those at a position of anyons. Hence, the quantum state $\ket{\Omega}$ or
the ground state of the Hamiltonian without quasiparticles can be
written as the following:
\bea
\ket{\Omega}\equiv\frac{1}{\sqrt{|G|}}\sum_{g\in G}g\ket{0}=\frac{1}{\sqrt{|G|}}\sum_{g\in G}g\ket{0_{L}}\otimes\ket{0_{R}},
\eea
where $|0_L\rangle$ and $|0_R\rangle$ are just shorthand expressions representing qubits live in a left region and a right region.
Therefore, the quantum state $|\psi\rangle$ on a disk manifold with holes can be written as the following:
\begin{align*}
\ket{\psi} & =\prod_{i_{L,v}}X_{i_{L,v}}\prod_{i_{L,c}}Z_{i_{L,c}}\prod_{i_{R,v}}X_{i_{R,v}}\prod_{i_{R,c}}Z_{i_{R,c}}\ket{\Omega}
\\
 & =\frac{1}{\sqrt{|G|}}\prod_{i_{L,v}}X_{i_{L,v}}\prod_{i_{L,c}}Z_{i_{L,c}}\prod_{i_{R,v}}X_{i_{R,v}}\prod_{i_{R,c}}Z_{i_{R,c}}\\
 & \times 
 \sum_{g\in G}g\ket{0_{L}}\otimes\ket{0_{R}}.
\end{align*}

For convenience, we define the operators as the followings:
\bea
E_{L}\equiv\prod_{i_{L,v}}X_{i_{L,v}}\prod_{i_{L,c}}Z_{i_{L,c}}, \qquad E_{R}\equiv\prod_{i_{R,v}}X_{i_{R,v}}\prod_{i_{R,c}}Z_{i_{R,c}}, 
\nn\\
\eea
\bea
E_L^2=1, \qquad E_R^2=1,
\eea
and write the element of the group $G$ as that $g\equiv g_{L}\otimes g_{R}$ and the wavefunction as that
\bea
\ket{\psi}=\frac{1}{\sqrt{|G|}}\sum_{g\in G}E_{L}g_{L}\ket{0_{L}}\otimes E_{R}g_{R}\ket{0_{R}}.
\eea
The operator $g_L$ and the operator $g_R$ do not contain a complete plaquette operators on an entangling surface. The operator $g_L$ only acts on one qubit for each plaquette operator and the action acts on three qubits for each plaquette operator on the entangling surface.

  Each $X$ and $Z$ operator has one ending point attached to an edge and the other ending point fixed
at where the corresponding anyon sits as in Fig. \ref{XZL}.
\begin{figure}
\begin{centering}
\includegraphics[scale=0.2]{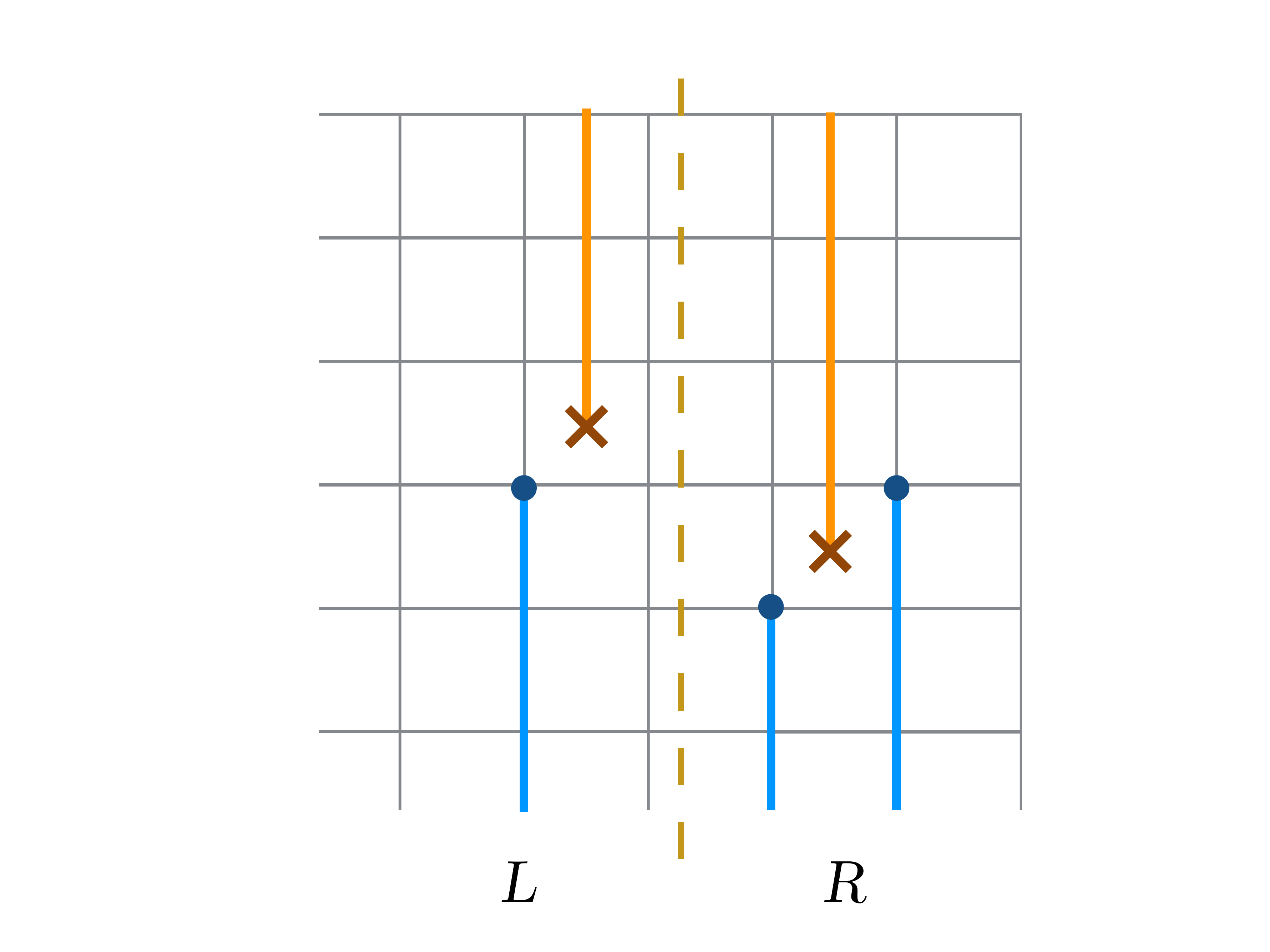}
\par\end{centering}
\caption{Condensation of each type of anyons is generated on the corresponding boundary, costing no energy to produce anyon at the boundary. The anyon in the bulk can be generated by a string operator with one end attached to
a suitable edge. Each blue line consists of a string of the $X$ operators, which
generates vortex particle pairs. The orange lines are strings of the $Z$ operators, which
generates charge particle pairs. Vortex particles condense at the rough boundary while charge particles condensed at the smooth boundary. By using the operation in Fig. \ref{fig:movecharge} and Fig. \ref{fig:movevortex}, we can move any configuration of blue lines and orange lines with a bulk end fixed and boundary ends at their proper boundary to the configuration shown in this Figure. Note that qubits live on links.
}
\label{XZL}
\end{figure}

Now we have the reduced density matrix of the $L$ system (the notations we use here are summarized in the Table. \ref{table:1} and the Table. \ref{table:2}):
\bea
\rho_{L}
&=&\mathrm{Tr}_{R}\ket{\psi}\bra{\psi}
\nn\\
&=&\frac{1}{|G|}\sum_{g,g'\in G}E_{L}{g_{L}}\ket{0_{L}}\bra{0_{L}}{g_{L}}'E_{L}
\bra{0_{R}}{g_{R}}'E_R E_R g_{R}\ket{0_{R}}
\nn\\
&=&\frac{1}{|G|}\sum_{g,g'\in G}E_{L}{g_{L}}\ket{0_{L}}\bra{0_{L}}{g_{L}}'E_{L}\cdot\bra{0_{R}}{g_{R}}'g_{R}\ket{0_{R}}
\nn\\
&=&\frac{1}{|G|}\sum_{g\in G, \tilde{g}\in G} E_{L}{g_{L}}\ket{0_{L}}\bra{0_{L}}g_L\tilde{g}_LE_{L}\cdot \bra{0_{R}}\tilde{g}_{R}\ket{0_{R}}
\nn\\
&=&
\frac{|G_{R}|}{|G|}\sum_{h\in G/G_{R},\tilde{g}\in G_{L}}E_{L}{h_L}\ket{0_{L}}\bra{0_{L}}h_L{\tilde{g}_{L}}E_{L},
\eea
where 
\bea
\tilde{g}_L\equiv g^{\prime}_Lg_L, \qquad \tilde{g}_R\equiv g^{\prime}_Rg_R.
\eea
Because the operator $g_L$ and the operator $g_R$ do not contain complete plaquette operators on the entangling surface, we remind our reader: $g_L\notin G_L$ and $g_R\notin G_R$, in which $G_L$ is a group generated by all the plaquette operators fully supported on a left region and $G_R$ is a group generated by all the plaquette operators fully supported on a right region. We also remind that the operator $\tilde{g}_L$ and the operator $\tilde{g}_R$ do not contain any plaquette operator on an entangling surface.

\begin{table}[h!]
\centering
\begin{tabular}{ |m{7em} | m{7cm}| } 

\hline
\textbf{Notation} & \textbf{Comment} \\ 
\hline
$G_L$ & The group generated by all the plaquette operators
 fully supported on the left region, as illustrated in 
Fig. \ref{fig:illustrategroups} (a). We also remind $g_L\notin G_L$.\\ 
\hline
$G_R$ & The group generated by all the plaquette operators
fully supported on the right region, as illustrated in
Fig. \ref{fig:illustrategroups} (b). We also remind $g_R\notin G_R$.\\ 
\hline
$G$ & The group generated by all the plaquette operators, 
as illustrated in Fig. \ref{fig:illustrategroups} (c).\\
\hline
$G/{G_R}$ & A quotient group.  
Each element is an equivalence class of elements of the group $G$ which have the same action on the left region, as illustrated in Fig. \ref{fig:illustratequotient}\\
\hline
$G/{(G_R \times G_L)}$ & A quotient group.  
The representative of this quotient 
group are generated by the plaquette operators on a 
border. The size of this group is $2^{n_L}$, where $n_L$ is a 
number of plaquette operators on a border. This is 
illustrated by Fig. \ref{fig:boundarygroup}.\\ 
\hline
\end{tabular}
\caption{List of notations of the groups.}
\label{table:1}
\end{table}

\begin{table}[h!]
\centering
\begin{tabular}{ |m{7em} | m{7cm}| } 

\hline
\textbf{Notation} & \textbf{Comment} \\ 
\hline
$g$ & An element of the group $G$. The element of the group
 $G$, the operator $g$ is defined as that $g\equiv g_L\otimes g_R$, in which the operator $g_L$ is supported on the 
left region and the operator $g_R$ is supported on the
 right region. The operator $g$ is related to the operator
$g'$
 by the relation $g'=g\tilde{g}$. We remind that the operator $g_L$ and
  the operator $g_R$ contain non-complete plaquette 
 operators.\\
\hline
$g'$ & An element of the group $G$. The operator $g^{\prime}$ is defined as that
$g'\equiv{g_L}'\otimes {g_R}'$, in which the operator ${g_L}'$ is supported on the left region and 
the 
operator ${g_R}'$ is supported on the right region. The
 operator $g$ is related to the operator $g'$ by the relation $g'=g\tilde{g}$. 
We remind that the operator $g^{\prime}_L$ and the operator 
$g^{\prime}_R$ contain non-complete plaquette operators.\\
\hline
$\tilde{g}$ & An element of the group $G$. The operator $\tilde{g}$ is defined as that $\tilde{g}\equiv{\tilde{g}_L}\otimes {\tilde{g}_R}$, in which 
the operator ${\tilde{g}_L}$ is supported on the left region, while 
the operator ${\tilde{g}_R}$ is supported on the right region. After 
the requirement that the operator is an identity 
operator on the right region ($\tilde{g}_R=\mathbb{I}$), the operator $\tilde{g}$
 becomes an element of $G_L$. \\
\hline
$h$ & An element of the quotient group $G/{G_R}$. We also 
define the operator $h\equiv h_L\otimes h_R$.\\ 
\hline
\end{tabular}
\caption{List of notations of the operators.}
\label{table:2}
\end{table}

\begin{figure}
\begin{centering}
\includegraphics[scale=0.2]{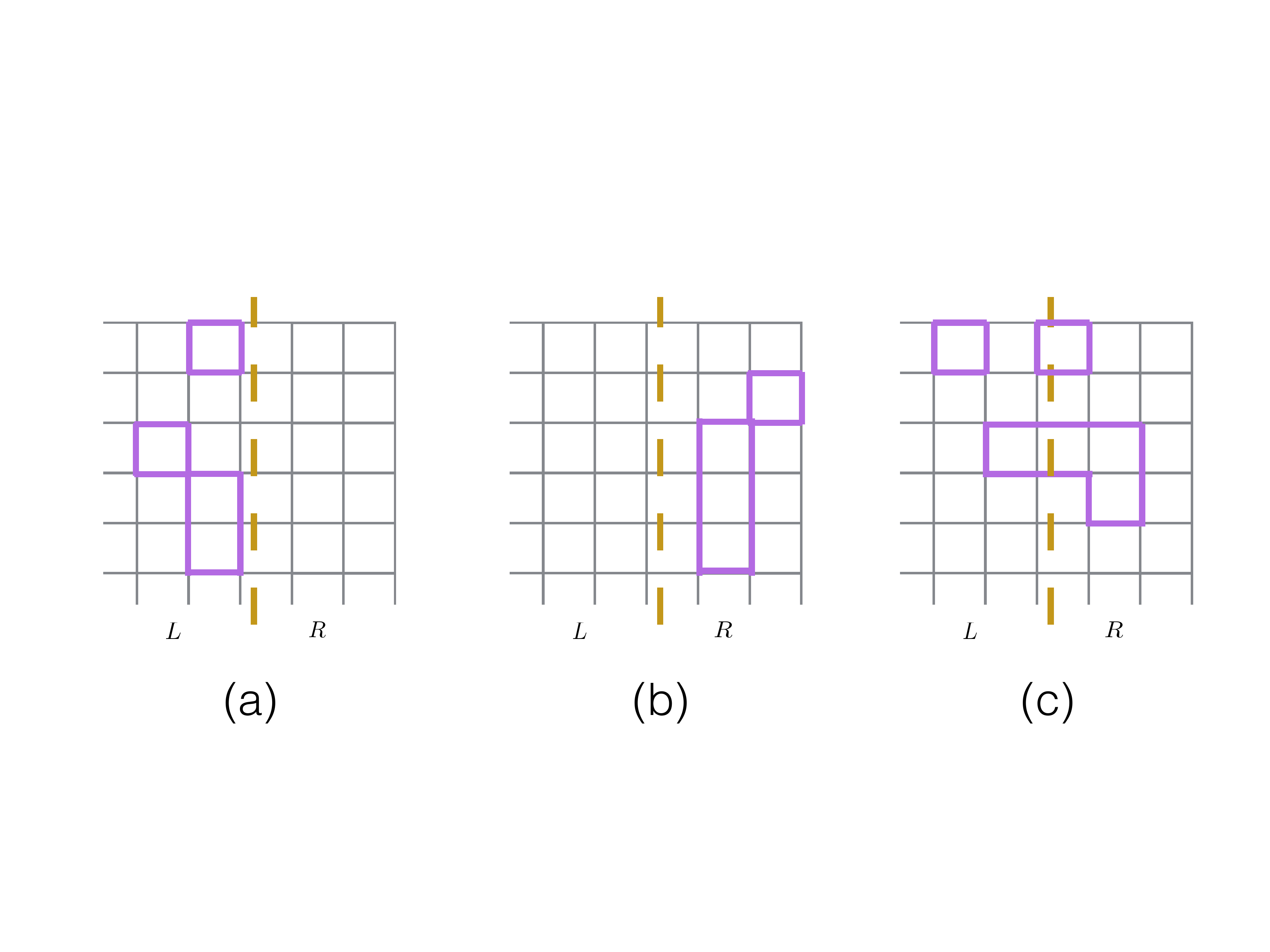}
\par\end{centering}
\caption{The illustration of groups $G$, $G_L$, and $G_R$.(a) The group $G_L$ is generated by plaquette operators supported on the left region. (b) The group $G_R$ is generated by plaquette operators supported on the left region. (c) The group $G$ is generated by plaquette operators on the whole region.
}
\label{fig:illustrategroups}
\end{figure}
\newpage 

\begin{figure}
\begin{centering}
\includegraphics[scale=0.2]{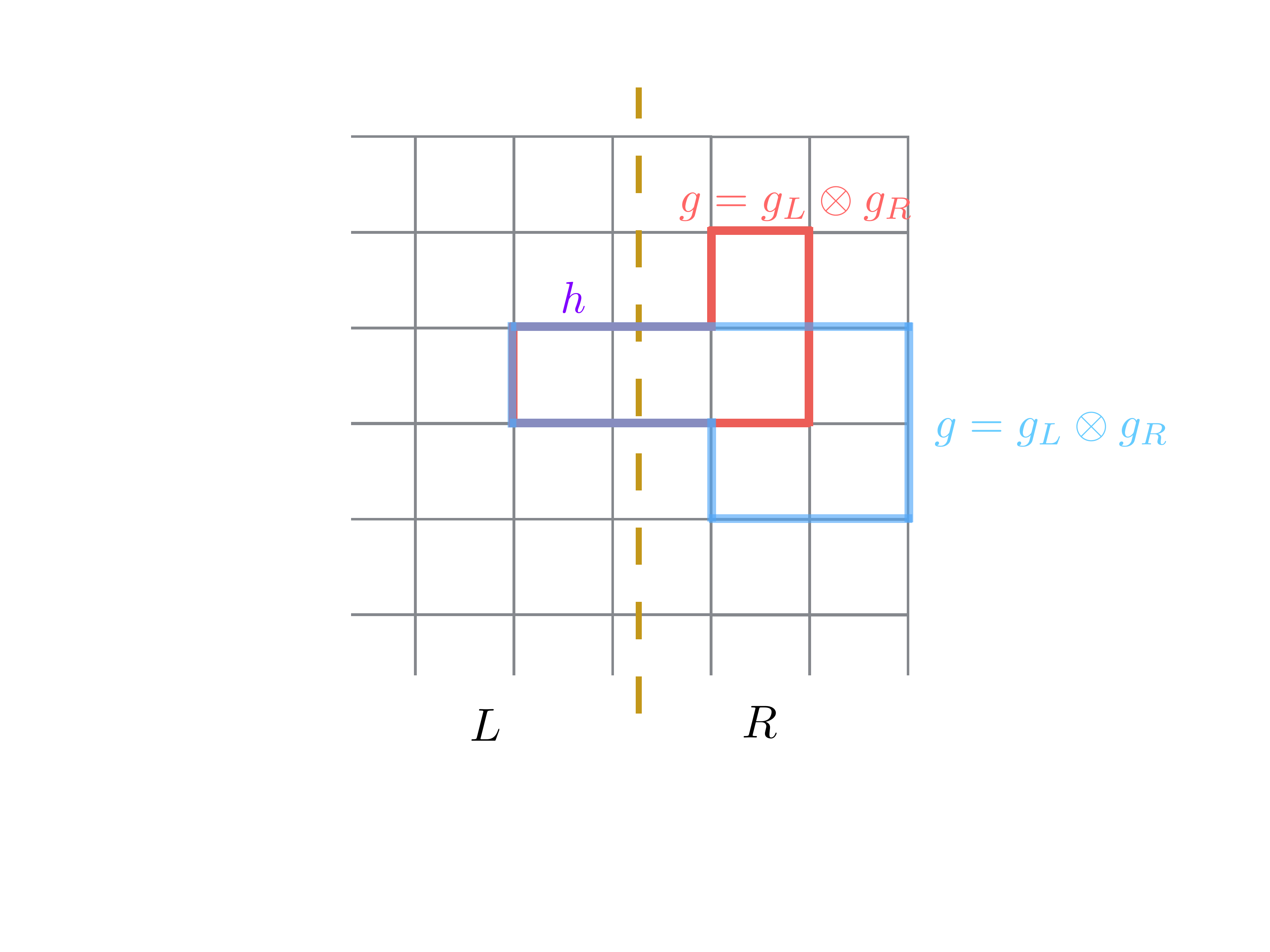}
\par\end{centering}
\caption{The illustration of the quotient group $G/G_R$. Two group elements in the group $G$ colored by red and blue having the same action on the left region are in the same equivalence class $h$ colored by the purple.
}
\label{fig:illustratequotient}
\end{figure}

\begin{figure}
\begin{centering}
\includegraphics[scale=0.2]{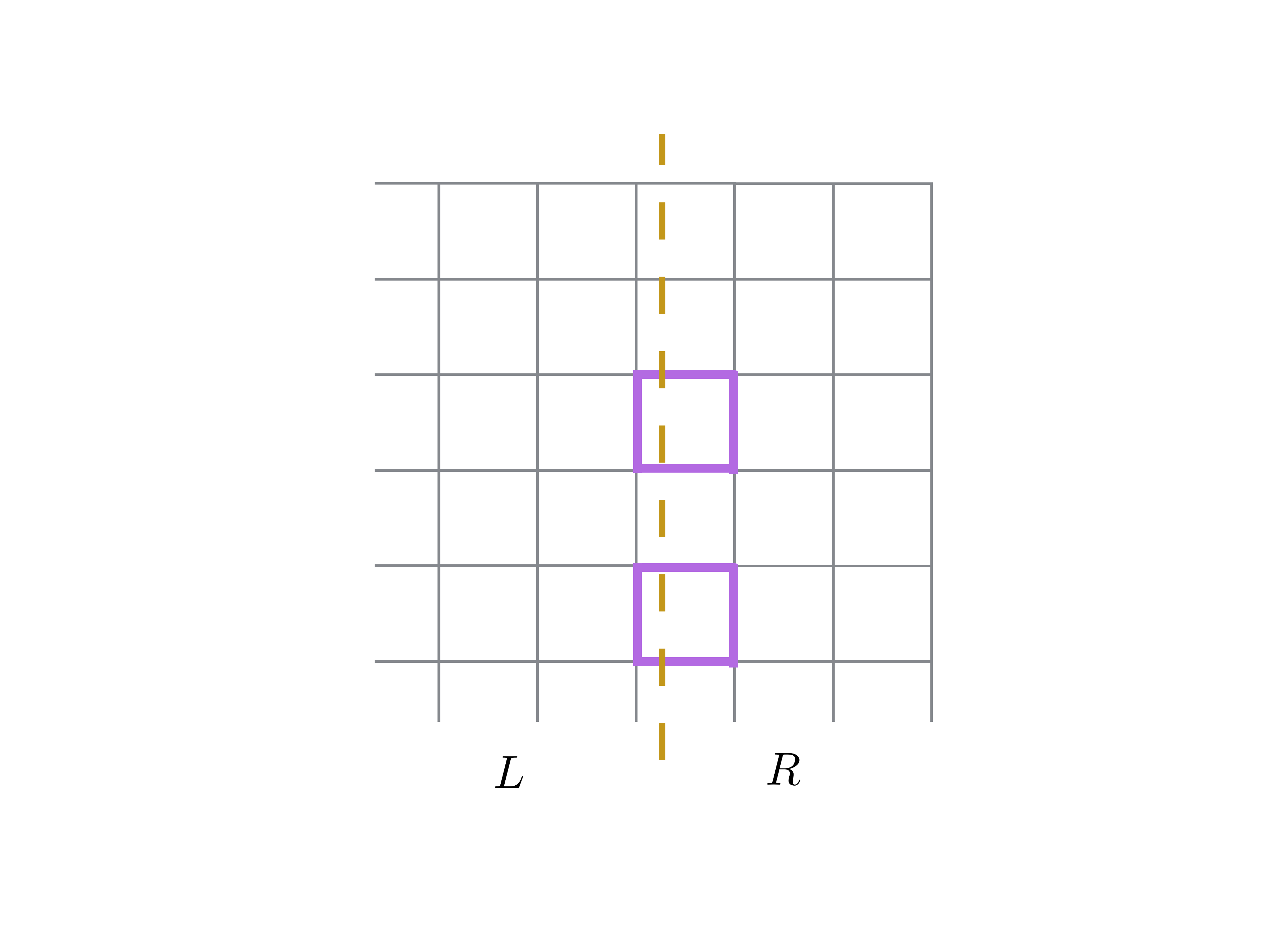}
\par\end{centering}
\caption{The illustration of the quotient group $G/(G_L \times G_R)$. The quotient group is generated by plaquette operators living on the border.
}
\label{fig:boundarygroup}
\end{figure}

We used the relation
\bea
\bra{0_R}\tilde{g}_R\ket{0_R}}={\delta_{\tilde{g}_{R}, \mathbb{I}}
\eea
in the third equality. Then we can obtain:
\bea
\rho_L^2
&=&\frac{|G_R|^2}{|G|^2}
\sum_{h, h'\in G/G_R, \tilde{g},\tilde{g}'\in G_L} 
E_Lh\ket{0_L}
 \bra{0_L}h{\tilde{g}_L}h'\ket{0_L} \bra{0_L}h'{\tilde{g}'_LE_L}
\nn\\
&=&\frac{|G_R|^2}{|G|^2} \sum_{h\in G/G_R, \tilde{g},\tilde{g}'\in G_L} E_Lh\ket{0_L} \bra{0_L}h{\tilde{g}_L}\tilde{g}'_LE_L
\nn\\
&=&\frac{|G_R|^2|G_L|}{|G|^2} \sum_{h\in G/G_R, \tilde{g}\in G_L} 
		E_Lh\ket{0_L} \bra{0_L}h{\tilde{g}_LE_L}
\nn\\
&=&\frac{|G_L||G_R|}{|G|}\rho_{L} 
\nn\\
&\equiv& \lambda \rho_L,
\eea
where
\bea
\lambda\equiv\frac{|G_{L}||G_{R}|}{|G|}=2^{-n_L}.
\eea

We used the relation
\bea
\bra{0_L}h{\tilde{g}_L}h'\ket{0_L}}={\delta_{h',h\tilde{g}_L}
\eea
in the second equality.
Thus, it is also easy to use the same way to obtain the equality
\bea
\rho_{L}^{n}=\lambda^{n}\left(\frac{\rho_{L}}{\lambda}\right).
\eea

The Rényi entropy of the order $\alpha$ in the toric code model on a disk manifold with any number of holes is:
\bea
S_{\alpha}\equiv\frac{1}{1-\alpha}\ln\mbox{Tr}\rho_L^{\alpha}=n_L\ln 2,
\eea
where $n_{L}$ is a number of plaquettes on the entangling surface between the region $L$ and the region $R$. 
In summary, entanglement entropy of the system with an arbitrary number of holes is:
\bea
S_{\mathrm{EE}, L}=S_{\alpha}=n_{L}\ln 2.
\eea
Note that at each hole, we can have two possibilities: presence and absence of an anyon. In this section, we consider the case where anyons are present at each hole. It turns out that all the other cases actually give the same Rényi entropy expression. Hence, the result does not give any constraint to a generalized concurrence.

\subsection{A Cylinder Manifold}
We consider the toric code model on a cylinder manifold, cutting it into two non-contractible sub-cylinders, region $A$ and region $B$. The operators on a boundary of the cylinder manifold satisfy the periodic boundary conditions for the upper and lower sides of the cylinder manifold. We also set the rough (smooth) boundary conditions for other two sides and put boundary conditions on an entangling surface to decompose a region \cite{Hung:2015fla}. When we cut the cylinder manifold into two sub-cylinder manifolds, the operators on an entangling surface between the region $A$ and the region $B$ in the sub-cylinder manifolds satisfy the rough (smooth) boundary condition.

 The toric code model on a cylinder manifold has two ground states \cite{Zhang:2011jd}:
\bea
|\psi_{00}\rangle
&=&\frac{1}{\sqrt{2N_q}}
\sum_{l=1}^{N_q}\big(|q_l=0,0{}_A\rangle|q_l=0,0{}_B\rangle
+|q_l=0,1{}_A\rangle|q_l=0, 1{}_B\rangle\big), 
\nn\\
|\psi_{01}\rangle
&=&\frac{1}{\sqrt{2N_q}}
\sum_{l=1}^{N_q}\big(|q_l=0,0{}_A\rangle|q_l=0,1{}_B\rangle
+|q_l=0,1{}_A\rangle|q_l=0, 0{}_B\rangle\big),
\eea
in which a crossing number to an entangling surface is always an even number, which is labeled by 0 at the first index of $q_l$, a winding number around the sub-cylinders $A$ or $B$ can be an even number, which is labeled by 0 at the second index of $q_l$, or 
an odd number, which is labeled by 0  at the second index of $q_l$, and $N_q=2^{n_L-1}$, where $n_L$ is a number of plaquettes on an entangling surface. A generic ground state is
\bea
|\psi\rangle=\alpha_{00}|\psi_{00}\rangle+\alpha_{01}|\psi_{01}\rangle.
\eea
Thus, the reduced density matrix in the region $A$ is
\bea
\rho_A
&=&\frac{1}{2N_q}
\sum_{l=1}^{N_q}\Bigg(\bigg(|\alpha_{00}|^2+|\alpha_{01}|^2\bigg)
\nn\\
&&\times
\big(|q_l=0,0{}_A\rangle\langle q_l=0,0{}_A|
+|q_l=0,1{}_A\rangle\langle q_l=0,1{}_A|\big)
\nn\\
&&+\bigg(\alpha_{00}^*\alpha_{01}+\alpha_{00}\alpha_{01}^*\bigg)
\nn\\
&&\times
\big(|q_l=0,0{}_A\rangle\langle q_l=0,1{}_A|
+|q_l=0,1{}_A\rangle\langle q_l=0,0{}_A|\big)\Bigg).
\nn\\
\eea
We can diagonalize the reduced density matrix of the region $A$ to obtain the eigenvalues of the reduced density matrix of the region $A$:
\bea
\frac{1}{2N_q}|\alpha_{00}+\alpha_{01}|^2, \qquad \frac{1}{2N_q}|\alpha_{00}-\alpha_{01}|^2,
\eea
and the Rényi entropy of the order $n$ is:
\bea
S_n&=&\frac{1}{1-n}\ln\Bigg(\bigg(\frac{1}{2N_q}\bigg)^nN_q\bigg(\sum_{i=1}^2(2p_i)^n\bigg)\Bigg)
\nn\\
&=&\ln N_q+\frac{1}{1-n}\ln\bigg(\sum_{i=1}^2p^n_i\bigg)
\nn\\
&=&n_L\ln 2-\Bigg(\ln 2-\frac{1}{1-n}\ln\bigg(\sum_{i=1}^2p_i^n\bigg)\Bigg),
\label{EQ 70}
\eea
where: 
\bea
p_1\equiv\frac{1}{2}|\alpha_{00}+\alpha_{01}|^2, \qquad p_2\equiv\frac{1}{2}|\alpha_{00}-\alpha_{01}|^2, 
\quad
 p_1+p_2=1.
\eea
Topological entanglement entropy is defined as a boundary independent part of entanglement entropy and from \eqref{EQ 70}, the topological entanglement entropy is
$\ln 2+\sum_{i=1}^2p_i\ln p_i$. The maximal topological 
entanglement entropy is $\ln2$. When $p_i=1/2$ for each index $i$ we can get the vanishing topological entanglement entropy:
\bea
\ln2-\sum_ip_i\ln p_i&=&\ln2-\frac{1}{2}\ln2-\frac{1}{2}\ln2=\ln2-\ln2=0.
\eea
Since we only have one independent parameter for $p_i$ from the second Renyi entropy in \eqref{EQ 70}, we have the following equations:
\bea
(n_L-1)\ln 2+\ln\mbox{Tr}\rho_A^2=\ln\big(2p_1^2-2p_1+1\big), 
\eea
\bea
2p_1^2-2p_1+\bigg(1-2^{n_L-1}\mbox{Tr}\rho_A^2\bigg)=0,
\eea
\bea
p_1&=&\frac{1\pm\sqrt{1-2\big(1-2^{n_L-1}\mbox{Tr}\rho_A^2}\big)}{2}, 
\nn\\
 p_2&=&\frac{1\mp\sqrt{1-2\big(1-2^{n_L-1}\mbox{Tr}\rho_A^2}\big)}{2}.
\label{EQ 76}
\eea
By comparing \eqref{EQ 28} and \eqref{EQ 76}, we find the entanglement entropy is quite similar to the entanglement entropy of the two qubits \cite{Bennett:1996gf}. 
We can know that the entanglement entropy should monotonically increase by decreasing $\mbox{Tr}\rho_A^2$ and also find that the generalized concurrence of the pure state in a toric code model on a cylinder manifold should be defined by
\bea
C(n_A, \psi)\equiv\sqrt{2\bigg(1-2^{n_A-1}\mbox{Tr}\rho_A^2\bigg)}.
\eea
If $n_A=1$, a definition of the generalized concurrence of the pure state goes back to the concurrence of two qubits \cite{Bennett:1996gf}. Even if $n_A=1$, the entanglement entropy of the toric model on the cylinder manifold does not vanish and has the contribution from classical Shannon entropy because the region $A$ is non-contractible. Hence, it is interesting to obtain the generalized concurrence of the pure state between different regions in the toric code model. We also find that the factor $2^{n_A-1}$ in the concurrence of the pure state depends on boundary degrees of freedom of the Hilbert space. Thus, a relation between the maximal violation of Bell's inequality and entanglement entropy possibly be expressed in terms of the generalized concurrence $C(n_A, \psi)$ with corresponding boundary degrees of freedom of the Hilbert space in the toric code model. This also motivates us to use the definition of the generalized concurrence.

\section{Some identities for computing the R-matrix}
\label{app2}

To compute the upper bound of the maximal violation of the Bell's inequality for the 2n-qubit states in Sec. \ref{4} , we need to compute elements of the R-matrix,
which have the form,
\bea
&&\mbox{Tr}\bigg(\big((a_1\otimes a_2\otimes\cdots\otimes a_n)\otimes(a_1\otimes a_2\otimes\cdots\otimes a_n)\big)
\nn\\
&&\times\big((b_1^T\otimes b_2^T\otimes\cdots\otimes b_n^T)\otimes(b_1^T\otimes b_2^T\otimes\cdots\otimes b_n^T)\big)
\nn\\
&&\times(\sigma_{i_1}\otimes\sigma_{i_2})\otimes(\sigma_{i_3}\otimes\sigma_{i_4})\otimes\cdots \otimes
(\sigma_{i_{2n-1}}\otimes\sigma_{i_{2n}})\bigg)
\nn\\
&=&\mbox{Tr}\bigg(\big((a_1b_1^T)\otimes(a_2b_2^T)\otimes\cdots\otimes(a_nb_n^T)\big)
\otimes \big((a_1b_1^T)\otimes(a_2b_2^T)\otimes\cdots\otimes(a_nb_n^T)\big)
\nn\\
&&\times(\sigma_{i_1}\otimes\sigma_{i_2})\otimes(\sigma_{i_3}\otimes\sigma_{i_4})\otimes\cdots\otimes
(\sigma_{i_{2n-1}}\otimes\sigma_{i_{2n}})\bigg)
\nn\\
&=&\mbox{Tr}\Bigg(\bigg(\big((a_1b_1^T)\otimes(a_2b_2^T)\big)\cdot(\sigma_{i_1}\otimes\sigma_{i_2})\bigg)
\otimes
\bigg(\big((a_3b_3^T)\otimes(a_4b_4^T)\big)\cdot(\sigma_{i_3}\otimes\sigma_{i_4})\bigg)\otimes\cdots
\nn\\
&&\otimes
\bigg(\big((a_{n}b_{n}^T)\otimes(a_{n}b_{n}^T)\big)\cdot(\sigma_{i_{2n-1}}\otimes\sigma_{i_{2n}})\bigg)\Bigg)
\nn\\
&=&\mbox{Tr}\Bigg(\bigg(\big((a_1b_1^T)\otimes(a_2b_2^T)\big)\cdot(\sigma_{i_1}\otimes\sigma_{i_2})\bigg)\Bigg)
\times
\mbox{Tr}\Bigg(\bigg(\big((a_3b_3^T)\otimes(a_4b_4^T)\big)\cdot(\sigma_{i_3}\otimes\sigma_{i_4})\bigg)\Bigg)\times\cdots
\nn\\
&&\times\mbox{Tr}\Bigg(\bigg(\big((a_{n-1}b_{n-1}^T)\otimes(a_{n}b_{n}^T)\big)\cdot(\sigma_{i_{2n-1}}
\otimes\sigma_{i_{2n}})\bigg)\Bigg)
\nn\\
&=&\mbox{Tr}\Bigg((\sigma_{i_1}\otimes\sigma_{i_2})\cdot\big((a_1b_1^T)\otimes(a_2b_2^T)\big)\Bigg)
\times
\mbox{Tr}\Bigg((\sigma_{i_3}\otimes\sigma_{i_4})\cdot\big((a_3b_3^T)\otimes(a_4b_4^T)\big)\Bigg)\times\cdots
\nn\\
&&\times\mbox{Tr}\Bigg((\sigma_{i_{2n-1}}\otimes\sigma_{i_{2n}})\cdot\big((a_{n-1}b_{n-1}^T)\otimes(a_{n}b_{n}^T)
\big)\Bigg)
\nn\\
&=&\mbox{Tr}\bigg(\big(\sigma_{i_1}a_1b_1^T\big)\otimes\big(\sigma_{i_2}a_2b_2^T\big)\bigg)
\times
\mbox{Tr}\bigg(\big(\sigma_{i_3}a_3b_3^T\big)\otimes\big(\sigma_{i_4}a_4b_4^T\big)\bigg)\times\cdots
\nn\\
&&\times\mbox{Tr}\bigg(\big(\sigma_{i_{2n-1}}a_{n-1}b_{n-1}^T\big)\otimes\big(\sigma_{i_{2n}}a_{n}b_{n}^T\big)
\bigg)
\nn\\
&=&\mbox{Tr}\bigg(\sigma_{i_1}a_1b_1^T\bigg)\times\mbox{Tr}\bigg(\sigma_{i_2}a_2b_2^T\bigg)
\times\cdots\times\mbox{Tr}\bigg(\sigma_{i_n}a_nb_n^T\bigg)
\nn\\
&&\times\mbox{Tr}\bigg(\sigma_{i_{n+1}}a_1b_1^T\bigg)
\times\mbox{Tr}\bigg(\sigma_{i_{n+2}}a_2b_2^T\bigg)\times\cdots
\times
\mbox{Tr}\bigg(\sigma_{i_{2n}}a_{2n}b_{2n}^T\bigg),
\eea
in which the notation $\otimes$ represents a tensor product of matrices while $\times$ and $\cdot$ stand for a matrix multiplication, a scalar multiplication or a dyad product, depending on the context. The two-component vectors $a_i$ and $b_j$ can be the following
\bea
\begin{pmatrix}1
 \\ 0
 \end{pmatrix} \ \mbox{or}\ \begin{pmatrix}0
 \\ 1
 \end{pmatrix}.
\eea
To obtain the eigenvalues of the matrix $R^{\dagger}R$, we list the following useful identities:
\bea
&&\mbox{Tr}\Bigg\lbrack\begin{pmatrix}0&1
 \\ 1&0
 \end{pmatrix} \begin{pmatrix}1
 \\0
 \end{pmatrix} \begin{pmatrix}0&1
 \end{pmatrix} \Bigg\rbrack=1, 
\quad
 \mbox{Tr}\Bigg\lbrack\begin{pmatrix}0&1
 \\ 1&0
 \end{pmatrix} \begin{pmatrix}0
 \\1
 \end{pmatrix} \begin{pmatrix}1&0
 \end{pmatrix} \Bigg\rbrack=1,
 \nn\\
 &&\mbox{Tr}\Bigg\lbrack\begin{pmatrix}0&-i
 \\ i&0
 \end{pmatrix} \begin{pmatrix}1
 \\0
 \end{pmatrix} \begin{pmatrix}0&1
 \end{pmatrix} \Bigg\rbrack=i, 
 \quad
 \mbox{Tr}\Bigg\lbrack\begin{pmatrix}0&-i
 \\ i&0
 \end{pmatrix} \begin{pmatrix}0
 \\1
 \end{pmatrix} \begin{pmatrix}1&0
 \end{pmatrix} \Bigg\rbrack=-i,
 \nn\\
 &&\mbox{Tr}\Bigg\lbrack\begin{pmatrix}1&0
 \\ 0&-1
 \end{pmatrix} \begin{pmatrix}1
 \\0
 \end{pmatrix} \begin{pmatrix}1&0
 \end{pmatrix} \Bigg\rbrack=1, 
 \quad
  \mbox{Tr}\Bigg\lbrack\begin{pmatrix}1&0
 \\ 0&-1
 \end{pmatrix} \begin{pmatrix}0
 \\1
 \end{pmatrix} \begin{pmatrix}0&1
 \end{pmatrix} \Bigg\rbrack=-1.
 \nn\\
\eea

\end{document}